\documentclass{LMCS}

\def\doi{8(3:26)2012}
\lmcsheading%
{\doi}
{1--32}
{}
{}
{Nov.~15, 2011}
{Sep.~29, 2012}
{}

\usepackage{amsthm,amssymb,amsmath}
\usepackage{ifthen}
\usepackage{graphicx}
\usepackage{hyperref}
\usepackage{xspace}
\usepackage{sfmath}
\usepackage{enumerate}

\newtheorem{df}[thm]{Definition}
\newtheorem{lemma}[thm]{Lemma}
\newtheorem{claim}[thm]{Claim}

\newenvironment{example}{\noindent{ Example:}}{}

\begin{document}

\newcommand{\lexi}{<_{\mathit{lex}}}
\newcommand{\koniec}{$\Box$}
\newcommand{\set}[1]{\{#1\}}
\newcommand{\Nat}{\mathbb N}

\newcommand{\bigpiece}{\mathcal T}
\newcommand{\miff}{\mbox{iff}}
\newcommand{\pen}{n}
\newcommand{\hole}{\Box}
\newcommand{\ltrivial}{\mathcal L}
\newcommand{\rtrivial}{\mathcal R}

\newcommand{\rlogic}{\mathcal R}
\newcommand{\llogic}{\mathcal L}

\newcommand{\tle}{\mathsf E}
\newcommand{\tlu}{\mathsf U}
\newcommand{\tlex}{\mathsf{EX}}
\newcommand{\tlspider}[1]{\mathsf{E}^{#1}\mathsf{X}}
\newcommand{\tlef}{\mathsf{EF}^{\downarrow}}
\newcommand{\tlleaf}{\mathsf{L}^{\downarrow}}
\newcommand{\tlfup}{\mathsf{F}^{\uparrow}}
\newcommand{\tlfh}{\mathsf{F}^{\rightarrow}}
\newcommand{\tlfhb}{\mathsf{F}^{\leftarrow}}
\newcommand{\tlag}{\mathsf{AG}}
\newcommand{\freef}[1]{#1^\Delta}
\newcommand{\tlsib}{\mathsf S}

\newcommand{\A}{\ensuremath{\mathbb{A}}\xspace}

\newcommand{\orderh}{<_h}
\newcommand{\orderfo}{<_{\text{dfs}}}

\newcommand{\mso}{\textup{MSO}\xspace}
\newcommand{\fo}{\textup{FO}\xspace}

\newcommand{\lgreen}{\le_\ltrivial}
\newcommand{\rgreen}{\le_\rtrivial}
\newcommand{\lgreeng}{\ge_\ltrivial}
\newcommand{\rgreeng}{\ge_\ltrivial}

\newcommand{\pieces}{\mathit{pieces}}
\newcommand{\subforest}{\leq}
\newcommand{\piece}{\preceq}
\newcommand{\pieceneq}{\prec}

\def\lca{closest common ancestor\xspace}
\def\lcas{closest common ancestors\xspace}
\def\lcapiece{cca-piece\xspace}
\def\lcapieces{cca-pieces\xspace}
\def\lcapiecewise{cca-piecewise\xspace}
\def\lcatame{cca-tame\xspace}
\newcommand{\wordpiece}{\sqsubseteq}

\title{Piecewise testable tree languages}

\author[M.~Boja{\'n}czyk]{Miko{\l}aj Boja{\'n}czyk\rsuper a}
\address{{\lsuper a}Warsaw University}
\email{bojan@mimuw.edu.pl}
\thanks{{\lsuper a}First author supported by Polish
    government grant no. N206 008 32/0810. This work was
    partially funded by the AutoMathA programme of the ESF and the PHC
    programme Polonium.}

\author[L.~Segoufin]{Luc Segoufin\rsuper b}
\address{{\lsuper b}INRIA and ENS-Cachan}
\email{luc.segoufin@inria.fr}

\author[H.~Straubing]{Howard Straubing\rsuper c}
\address{{\lsuper c}Boston College}
\email{Straubing@cs.bc.edu}
\thanks{{\lsuper c}Third author supported by National Science Foundation grant  CCF-0915065}

\keywords{First-order logic on trees,algebra}
\subjclass{F.4.3,F.4.1}

\begin{abstract}
  This paper presents a decidable characterization of tree languages
  that can be defined by a boolean combination of $\Sigma_1$
  sentences. This is a tree extension of the Simon theorem, which says
  that a string language can be defined by a boolean combination of
  $\Sigma_1$ sentences if and only if its syntactic monoid is
  ${\mathcal J}$-trivial.
\end{abstract}

\maketitle

\section{Introduction}

Logics for expressing properties of labeled trees and forests figure
importantly in several different areas of Computer Science.  
This paper is about logics on finite trees. All the logics we consider are
less expressive than monadic second-order logic, and thus can
 be captured by finite automata
on finite trees. Even with these restrictions, this encompasses  a large body
of important logics, such as variants of first-order logic, temporal logics
including CTL* or CTL, as well as query languages used in XML. 

One way of trying to understand a logic is to give an effective
characterization. An effective characterization for a logic $\mathcal L$ is an
algorithm which inputs a tree automaton, and says if the language recognized by
the automaton can be defined by a sentence of the logic $\mathcal L$.  Although
giving an effective characterization may seem an artificial criterion for
understanding a logic, it has proved to work very well, as witnessed by decades
of research, especially into logics for words.  In the case of words, effective
characterizations have been studied by applying ideas from algebra: A property
of words over a finite alphabet $A$ defines a set of words, that is a language
$L\subseteq A^*.$ As long as the logic in question is no more expressive than
monadic second-order logic, $L$ is a regular language, and definability in the
logic often boils down to verifying a property of the {\it syntactic monoid} of
$L$ (the transition monoid of the minimal automaton of $L$). This approach
dates back to the work of McNaughton and Papert~\cite{mcnaughton} on
first-order logic over $<$ (where $<$ denotes the usual linear ordering of
positions within a word).  A comprehensive survey, treating many extensions and
restrictions of first-order logic, is given by
Straubing~\cite{straubing}. Th\'erien and
Wilke~\cite{wilke,therienwilkefo2,untilcompute} similarly study temporal logics
over words.

An important early discovery in this vein, due to Simon~\cite{simonpiecewise},
treats word languages definable in first-order logic over $<$ with low
quantifier complexity. Recall that a $\Sigma_1$ sentence is one that uses only
existential quantifiers in prenex normal form, e.g.~$\exists x \exists y\ x
<y$. Simon proved that a word language is definable by a boolean combination of
$\Sigma_1$ sentences over $<$ if and only its syntactic monoid $M$ is {\it
  ${\mathcal J}$-trivial}.  This means that for all $m,m'\in M,$ if $MmM=Mm'M,$ then $m=m'.$ (In
other words, distinct elements generate distinct two-sided semigroup ideals.)
Thus one can effectively decide, given an automaton for $L,$ whether $L$ is
definable by such a sentence.  (Simon did not discuss logic {\it per se}, but
phrased his argument in terms of {\it piecewise testable languages} which are exactly those definable by boolean
combinations of $\Sigma_1$ sentences.)

There has been some recent success in extending these methods to trees and
forests. (We work here with unranked trees and forests, and not binary or
ranked ones, since we believe that the definitions and proofs are cleaner in
this setting.) The algebra is more complicated, because there are two
multiplicative structures associated with trees and forests, both horizontal
and a vertical concatenation. Benedikt and Segoufin~\cite{segoufinfo} use these
ideas to effectively characterize sets of trees definable by first-order logic
with the parent-child relation. Boja\'nczyk~\cite{fo2tree} gives a decidable
characterization of properties definable in a temporal logic with unary
ancestor and descendant operators. Similarly Boja\'nczyk and
Segoufin~\cite{BoSegDelta2} and Place and Segoufin~\cite{PlaSegfo2} provided
decidable characterizations of tree languages definable in $\Delta_2(<)$ and
$FO_2(<,<_h)$ where $<$ denotes the descendant-ancestor relationship while $<_h$
denotes the sibling relationship. The general theory of the `forest algebras'
that underlie these studies is presented by Boja\'nczyk and
Walukiewicz~\cite{forestalgebra}.

In the present paper we provide a further illustration of the utility of these algebraic methods by generalizing Simon's theorem from words to trees.  In fact, we give several such generalizations, differing in the kinds of atomic formulas we allow in our $\Sigma_1$ sentences.

In Section 2 we present our basic terminology concerning trees, forests, and logic.  Initially our logic contains two orderings:  the ancestor relation between nodes in a forest, and the depth-first, left-first, total ordering of the nodes of a forest.  In Section 3 we describe the algebraic apparatus.  This is the theory of forest algebras developed in ~\cite{forestalgebra}.

In Section 4 we give our main result, an effective test of whether a given language is piecewise testable (Theorem 4.) The test consists of verifying that the syntactic forest algebra satisfies a particular identity.  While we have to some extent drawn on
Simon's original argument, the added complexity of the tree setting makes both
formulating the correct condition and generalizing the proof quite nontrivial.
We give a quite different, equivalent identity in Proposition 18, which makes clear the precise relation between piecewise testability for forest languages and ${\mathcal J}$-triviality.

In Section 5, we study in detail a variant of our logic in which the binary ancestor relation is replaced by a ternary closest common ancestor relation, and prove a version of our main theorem for this case.  Section 6 is devoted to other variants:  the far simpler case of languages defined by $\Sigma_1$ sentences (instead of boolean combinations thereof); the logics in which only the ancestor relation is present, and in which the horizontal ordering on siblings is present; and, since our algebraic formalism concerns forests rather than trees, the modifications necessary to obtain an effective characterization of the piecewise testable tree languages.  We discuss some directions for further research in the concluding Section 7.

An earlier, much abbreviated version of this paper, without complete proofs, was presented at the 2008 IEEE Symposium on Logic in Computer Science.

\section{Notation}
\label{sec:notation}

\paragraph*{Trees, forests and contexts.}
In this paper we work with finite unranked ordered trees and forests over a
finite alphabet \A. Formally, these are expressions defined inductively as
follows: for any $a\in\A$, $a$ is a tree. If $t_1,\ldots,t_n$ is a finite
sequence of trees, then $t_1+\cdots+t_n$ is a forest. If $s$ is a forest and
$a\in \A$, then $as$ is a tree. It will also be convenient to have an
\emph{empty forest}, that we will denote by 0, and this forest is such that
$a0=a$ and $0+t=t+0=t$. Forests and trees alike will be denoted by the letters
$s,t,u,\ldots$

For example, the forest that we conventionally draw as 
\medskip
\begin{center}
  \includegraphics[scale=1]{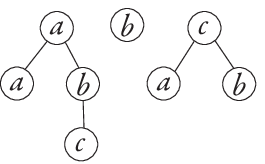}
\end{center}
corresponds to the expression
\begin{eqnarray*}
  t=a(a+bc)+b+c(a+b)\ .
\end{eqnarray*}
When there is no ambiguity we use $as$ instead of $a(s)$. In particular $bc$
stands for the tree whose root has label $b$ and has a unique child of label $c$.

The notions of node, child, parent, descendant and ancestor relations between
nodes are defined in the usual way. We write $x<y$ to say that $x$ is a strict
ancestor of $y$ or, equivalently, that $y$ is a strict descendant of $x$.  We
say that a sequence $y_1,\ldots,y_n$ of nodes forms a \emph{chain} if we have
$y_i < y_{i+1}$ for all $1\leq i<n$.  As our forests are ordered, each forest
induces a natural linear order on its set of nodes that we call the
\emph{forest-order} and denote by $\orderfo$, which corresponds to the
depth-first left-first traversal of the forest or, equivalently, to the order
provided by the expression denoting the forest seen as a word. We write
$\orderh$ for the \emph{horizontal-order}, i.e. $x \orderh y$ expresses the
fact that $x$ is a sibling of $y$ occurring strictly before $y$ in the
forest-order. Finally, the \emph{\lca} of two nodes $x,y$ is the unique node
$z$ that is a descendant of all nodes that are ancestors of both $x$ and $y$.

If we take a forest and replace one of the leaves by a special symbol $\hole$,
we obtain a \emph{context.} This special node is called the
\emph{hole} of the context. Contexts will be denoted using letters $p,q,r$. 
For example, from the forest $t$ given above, we can obtain, among others, the context
\begin{eqnarray*}
  p=a(a+bc)+b+c(\hole+b)\ .
\end{eqnarray*}

A forest $s$ can be substituted in place of the hole of a context $p$; the
resulting forest is denoted by $ps$. If we take the context $p$ above and if
$s=(b+ca),$ then  
\begin{eqnarray*}
    ps=a(a+bc)+b+c(b+ca+b)\ .
\end{eqnarray*}
This is depicted in the figure below.

\begin{center}
  \includegraphics[scale=1]{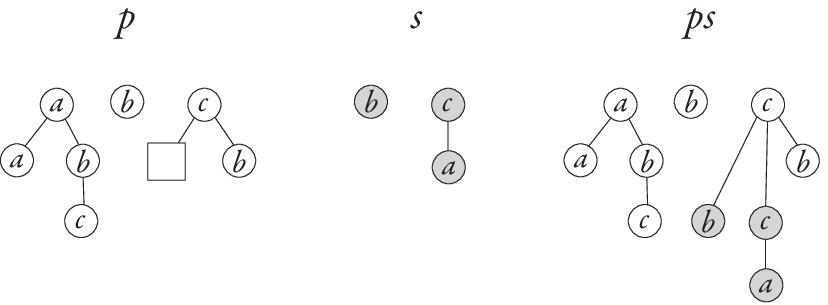}
 \end{center}

 There is a natural composition operation on contexts: the
context $qp$ is formed by replacing the hole of $q$ with $p$.
This operation is associative, and satisfies $(pq)s=p(qs)$ for all
forests $s$ and contexts $p$ and $q$.

We distinguish a special context, \emph{the empty context}, denoted $\hole$. It
satisfies $\hole s=s$ and $\hole p=p\hole=p$ for any forest $s$ and context $p$.

\paragraph*{Regular forest languages.}
A set $L$ of forests over \A is called a \emph{forest language.}  There are
several notions of automata for unranked ordered trees, see for
instance~\cite[chapter 8]{tata}. They all recognize the same class of
forest languages, called \emph{regular}, which also corresponds to definability in \mso as
defined below.

\paragraph*{Piecewise testable languages.}

We say that a forest $s$ \emph{is a piece} of a forest $t$ if there is an
injective mapping from nodes of $s$ to nodes of $t$ that preserves the
label of the node together with the forest-order and the ancestor relationship.  An equivalent definition is that
the piece relation is the reflexive transitive closure of the relation
\begin{eqnarray*}
\set{  (pt, pat) : \mbox{$p$ is a context, $a$ is a node, $t$ is a forest or empty}}
\end{eqnarray*}
In other words, a piece of $t$ is obtained by removing nodes from $t$ while
preserving the forest-order and the ancestor relationship.  We write $s \piece
t$ to say that $s$ is a piece of $t$.  In the example above, $a(a+b)+c$ is a
piece of $t$.

We extend the notion of piece to contexts. In this case, the hole must
be preserved while removing the nodes:

\begin{center}
  \includegraphics[scale=1]{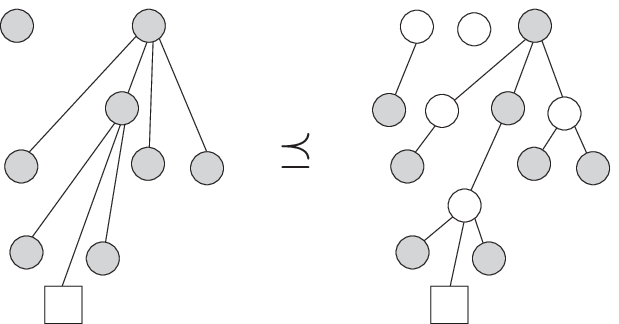}
 \end{center}
 The size of a piece is the size of the corresponding forest, i.e.~the number
 of its nodes.  The
 notions of piece for forests and contexts are related, of course. For
 instance, if $p$, $q$ are contexts with $p \piece q$, then $p0 \piece
 q0$. Also, conversely, if $s \piece t$, then there are contexts $p \piece q$
 with $s=p0$ and $t=q0$.

 A forest language $L$ over \A is called \emph{piecewise testable} if
 there exists $n\geq 0$ such that membership of $t$ in $L$ is
 determined by the set of pieces of $t$ of size $n$ or less.  Equivalently, $L$ is a finite boolean combination of languages
 $\{t:s\piece t\},$ where $s$ is a forest.  Every piecewise testable
 forest language is regular, since given $n\geq 0,$ a finite
 automaton can calculate on input $t$ the set of pieces of $t$ of
 size no more than~$n.$

\paragraph*{Logic.}
Regularity and piecewise testability correspond to definability in a logic,
which we now describe.  A forest can be seen as a logical relational structure.
The domain of the structure is the set of nodes. The signature contains a unary
predicate $P_a$ for each symbol $a$ of the label alphabet \A, plus possibly
some extra predicates on nodes, such as the descendant relationship, the
forest-order or the closest common ancestor.  Let $\Omega$ be a set of
predicates.  The predicates $\Omega$ that we use always include $(P_a)_{a \in
  \Sigma}$ and equality, hence we do not explicitly mention them in the
sequel. We use the classical syntax and semantics for first-order logic,
$\fo(\Omega)$, and monadic second order logic, $\mso(\Omega)$, building on the
predicates in $\Omega$. Given a sentence $\phi$ of any of these formalisms, the
set of forests that are a model for $\phi$ is called \emph{the language defined by
$\phi$}.
 In particular a language is definable in
$\mso(<,\orderh)$ iff it is regular~\cite[chapter 8]{tata}.

A $\Sigma_1(\Omega)$ formula is a formula $\exists x_1\cdots x_n~\gamma$, where
the formula $\gamma$ is quantifier-free and uses predicates from $\Omega$.
Initially we will consider two predicates on nodes: the ancestor order $x < y$
and the forest-order $x \orderfo y$.  Later on, we will see other combinations
of predicates, for instance when the closest common ancestor is added, and the
forest-order is removed.

It is not too hard to show that a forest language $L$ can
be defined by a $\Sigma_1(<,\orderfo)$ sentence if and only if it is closed
under adding nodes, i.e.
\begin{eqnarray*}
  pt \in L \qquad \Rightarrow \qquad pqt \in L 
\end{eqnarray*}
holds for all contexts $p$, $q$ and forests $t$. Moreover this condition can be
effectively decided given any reasonable representation of the language $L$. We
will carry out the details in Section~\ref{sec-sigma1}.

We are more interested here in the  boolean 
combinations of properties definable in $\Sigma_1(<,\orderfo)$. It is easy to
see that:

\begin{prop}\label{boolean_sigma_1}
 A forest language is piecewise testable 
  iff it is definable by a boolean combination of $\Sigma_1(<,\orderfo)$ sentences.
\end{prop}
One direction is immediate as for any forest $s$, the set of forests having $s$ as a
piece is easily definable in $\Sigma_1(<,\orderfo)$. For instance the sentence
\begin{eqnarray*}
  \exists x,y,z,u~~P_a(x) \land P_a(y) \land P_b(z) \land P_c(u) \land x<y
  \land x< z \land y \orderfo  z \land \lnot (x<u) \land x \orderfo u
\end{eqnarray*}
defines the language of forests having $a(a+b)+c$ as a piece.

For the other direction, notice that for any language definable in
$\Sigma_1(<,\orderfo)$, by disambiguating the relative positions between
each pair of variables, one can compute a finite set of pieces such that a forest belongs
to the language iff it has one of them as a piece.
For instance the sentence
\begin{eqnarray*}
  \exists x,y,z,u~~P_a(x) \land P_a(y) \land P_b(z) \land P_c(u) \land x<y
  \land x< z \land y \orderfo  z \land \lnot (x<u)
\end{eqnarray*}
defines the language of forests having $a(a+b)+c$, $c + a(a+b)$ or $ca(a+b)$ as a piece.

\medskip

This result does not address the question of effectively determining whether a given regular forest language admits either of these equivalent descriptions. Such an effective characterization is
the goal of this paper:

\paragraph*{The problem.}

Find an algorithm that decides whether or not a given regular forest language is
piecewise testable.

\medskip As noted in the introduction, the corresponding problem for words was
solved by Simon, who showed that a word language $L$ is piecewise
testable if and only if its syntactic monoid $M(L)$ is ${\mathcal
  J}$-trivial~\cite{simonpiecewise}; that is, if distinct elements $m,m'$ always generate distinct two-sided ideals.  Note that one can test, given the multiplication table of 
  a finite monoid $M,$ whether $M$ is $\mathcal
J$-trivial in time polynomial in $|M|$: for each $m \neq m' \in M$, one
calculates the ideals $MmM$ and $Mm'M$ and then verifies that they are
different.
Therefore, it is decidable if a given regular word language is
piecewise testable. We assume that the language $L$ is given by its
syntactic monoid and syntactic morphism, or by some other
representation, such as a finite automaton, from which these can be
effectively computed.

We will show that a similar characterization can be found for forests;
although the characterization will be more involved.  For decidability, it
is not important how the input language is represented. In this paper,
we will represent a forest language by a morphism into a finite forest algebra that
recognizes it.  Forest algebras are described in the next section.

\section{Forest algebras}\label{sec:forest-algebras}

\paragraph*{Forest algebras.} Forest algebras were introduced by Boja{\'n}czyk
and Walukiewicz as an algebraic formalism for studying regular tree
languages~\cite{forestalgebra}.  Here we give a brief summary of the definition of
these algebras and their important properties.  A forest algebra consists of a
pair $(H,V)$ of monoids, subject to some additional requirements, which we
describe below.  We write the operation in $V$ multiplicatively and the
operation in $H$ additively, although $H$ is not assumed to be commutative.  We
denote the identity of $V$ by $\hole$ and that of $H$ by 0.
 
 We require that $V$ act on the left of $H$.  That is, there is a map
 $$(h,v)\mapsto vh\in H$$
 such that 
 $$w(vh)=(wv)h$$
 for all $h\in H$ and $v,w\in V.$ We further require that this action be {\it
   monoidal,} that is,
 $$\hole\cdot h=h$$
 for all $h\in H,$ and that it be {\it faithful},
that is, if $vh=wh$ for all $h\in H,$ then $v=w.$
 
 We further require that for every $g\in H,$ $V$ contains elements $(\hole+g)$
 and $(g+\hole)$ such that
 $$(\hole+g)h=h+g, (g+\hole)h=g+h$$
 for all $h\in H.$  Observe, in particular, that for all $g,h\in H,$
 $$(g+\hole)(h+\hole)= (g+h)+\hole,$$
 so that the map $h\mapsto h+\hole$ is a morphism embedding $H$ as a submonoid of $V.$
 
 A morphism $\alpha:(H_1,V_1)\to (H_2,V_2)$ of forest algebras is actually a
 pair $(\gamma,\delta)$ of monoid morphisms $\gamma: H_1 \to H_2$, $\delta: V_1
 \to V_2$ such that $\gamma(vh)=\delta(v)\gamma(h)$ for all $h\in H,$ $v\in V.$
 However, we will abuse notation slightly and denote both component maps by
 $\alpha.$
 
 Let \A be a finite alphabet, and let us denote by $H_{\A}$ the set of forests
 over \A, and by $V_{\A}$ the set of contexts over \A. Clearly $H_{\A}$ forms a
 monoid under $+,$ $V_{\A}$ forms a monoid under composition of contexts (the
 identity element is the empty context $\hole$), and substitution of a forest
 into a context defines a left action of $V_{\A}$ on $H_{\A}.$ It is straightforward
 to verify that this action makes $(H_{\A},V_{\A})$ into a forest algebra, which we
 denote $\A^{\Delta}.$ If $(H,V)$ is a forest algebra, then every map $f$ from
 \A to $V$ has a unique extension to a forest algebra morphism
 $\alpha:\A^{\Delta}\to (H,V)$ such that $\alpha(a \hole)=f(a)$ for all $a\in
 \A.$ In view of this universal property, we call $\A^{\Delta}$ the {\it free
   forest algebra} on \A.
 
 We say that a forest algebra $(H,V)$ {\it recognizes} a forest language
 $L\subseteq H_{\A}$ if there is a morphism $\alpha:\A^{\Delta}\to (H,V)$ and a
 subset $X$ of $H$ such that $L=\alpha^{-1}(X).$ We also say that the morphism
 $\alpha$ recognizes $L.$ It is easy to show that a forest language is regular
 if and only if it is recognized by a \emph{finite} forest algebra.

 Given $L\subseteq H_{\A}$ we define an equivalence relation $\sim_L$ on
 $H_{\A}$ by setting $s\sim_L s'$ if and only if for every context $p\in
 V_{\A},$ $ps$ and $ps'$ are either both in $L$ or both outside of $L.$ We
 further define an equivalence relation on $V_{\A}$, also denoted $\sim_L,$ by
 $p\sim_L p'$ if for all $s\in H_{\A},$ $ps\sim_L p's.$ This pair of
 equivalence relations defines a congruence of forest algebras on
 $\A^{\Delta}.$ The quotient $(H_L,V_L)$ is called the {\it syntactic forest
   algebra} of $L.$ The projection morphism of $\A^{\Delta}$ onto $(H_L,V_L)$
 is denoted $\alpha_L$ and called the {\it syntactic morphism} of
 $L$. $\alpha_L$ always recognizes $L$ and it is easy to show that $L$ is
 regular iff $(H_L,V_L)$ is finite.

\paragraph*{Idempotents and aperiodicity.}
We recall the well known notions of idempotent and aperiodicity.  If $M$ is a
finite monoid and $m\in M,$ then there is a unique element $e=m^n,$ where
$n>0,$ such that $e$ is \emph{idempotent,} \emph{i.e.,} $e^2=e.$ If we take a
common multiple of these exponents $n$ over all $m\in M,$ we obtain an integer
$\omega>0$ such that $m^{\omega}$ is idempotent for every $m\in M.$ Observe
that while infinitely many different values of $\omega$ have this property with
respect to $M,$ the value of $m^{\omega}$ is uniquely determined for each $m\in
M.$

Let $(H,V)$ be a forest algebra.  Since we write the operation in $H$ additively, we denote powers of $h\in H$ by $n\cdot h,$ where $n\geq 0.$  As noted above, $H$ embeds in $V,$ so any $\omega>0$ that yields idempotents for $V$ serves as well for $H.$  That is, there is an integer $\omega>0$ such that $v^{\omega}$ is idempotent for all $v\in V,$ and $\omega\cdot h$ is idempotent for all $h\in H.$

We say that a finite monoid $M$ is {\it aperiodic} if it contains no nontrivial groups.  Since the set of elements of the form $m^{\omega}m^k$ for $k\geq 0$ is a group, aperiodicity is equivalent to having $m^{\omega} = m^{{\omega}+1}$ for all $m\in M.$  In this case we can take $\omega = |M|.$ All the finite monoids that we encounter in this paper are aperiodic.  In particular, every ${\mathcal J}$-trivial monoid is aperiodic, because all elements of a group in a finite monoid generate the same two-sided ideal.

\paragraph*{Pieces.}
Recall that in Section~\ref{sec:notation}, we defined the piece relation for
contexts in the free forest algebra. We now extend this definition to an
arbitrary forest algebra $(H,V)$. The general idea is that a context $v \in V$
is a piece of a context $w \in V$, denoted by $v \piece w$, if one can
construct a term (using elements of $H$ and $V$) which evaluates to $w$, and
then take out some parts of this term to get $v$.

Let $(H,V)$ be a forest algebra. We say $v \in V$ \emph{is a piece}
  of $w \in V$, denoted by $v \piece w$, if $\alpha(p)=v$ and
  $\alpha(q)=w$ hold for some morphism
  \begin{eqnarray*}
    \alpha : \freef \A \to (H,V)
  \end{eqnarray*}
  and some contexts $p \piece q$ over \A.  The relation $\piece$ is
  extended to $H$ by setting $g \piece h$ if $g=v0$ and $h=w0$ for
  some contexts $v \piece w$.

As we will see in the proof of Lemma~\ref{lemma:equivalent-def}, in
the above definition, we can replace the term ``some morphism'' by
``any surjective morphism''.  The following example shows that
although the piece relation is transitive in the free algebra $\freef
\A$, it may no longer be so in a finite forest algebra.

\medskip
\begin{example}
  Consider the syntactic algebra of the language $\set{abcd}$, which
  contains only one forest, which in turn has just one path, labeled
  by $abcd$. The context part of the syntactic algebra has twelve
  elements: an error element $\infty$, and one element for each infix
  of $abcd$.  We have
  \begin{eqnarray*}
    a \piece aa = \infty  = bd  \piece bcd
  \end{eqnarray*}
 but we do not
  have $a \piece bcd$.
\end{example}
\medskip

We will now show that in a finite forest algebra, one can compute the
relation $\piece$ in time polynomial in $|V|$. The idea is to use a different
but equivalent definition. Let $R$ be the smallest relation on $V$
that satisfies the following rules, for all $v,v',w,w' \in V$:
\begin{eqnarray*}
  \begin{array}{rcll}
      \hole & R &  v\\
  v & R & v\\
  vw & R & v'w' \qquad & \mbox{ if $v\ R\ v'$ and $w\ R\ w'$}\\
  \hole + v0 & R & \hole +v'0 & \mbox{ if $v\ R\ v'$}\\
   v0 + \hole & R & v'0 + \hole  & \mbox{ if $v\ R\ v'$}
  \end{array}
\end{eqnarray*}

\begin{lemma}\label{lemma:equivalent-def}
  Over any finite forest algebra the relations $R$ and $\piece$ are the same.
\end{lemma}
In any finite algebra, the relation $R$ can be computed by applying
the rules until no new relations can be added. This gives the
following corollary:
\begin{cor}\label{cor:decidable}
  In any given finite forest algebra, the relation $\piece$ on contexts (also
  on forests) can be calculated in polynomial time.
\end{cor}

\medskip
\begin{proof}[Proof of Lemma~\ref{lemma:equivalent-def}]
  We first show the inclusion of $R$ in $\piece$. Let $\alpha : \freef \A \to
  (H,V)$ be any surjective morphism.  A simple induction on the number of steps
  used to derive $v \ R \ w$, produces contexts $p \piece q$ with $\alpha(p)=v$
  and $\alpha(q)=w$. The surjectivity of $\alpha$ is necessary for starting the
  induction in the case $\hole\ R\ v$. 

  For the opposite inclusion, suppose $v\piece w.$ Then there is a morphism
  $\alpha : \freef \A \to
  (H,V)$ and contexts $p\piece q$ such that $v=\alpha(p),$ $w=\alpha(q).$
  We will show that $\alpha(p)\ R \
  \alpha(q)$  by
  induction on the size of $p$:
  \begin{iteMize}{$\bullet$} 
  \item If $p$ is the empty context, then the result follows thanks to
    the first rule in the definition of $R$.  If $p=a\hole$ then from $p
    \piece q$ it follows that $q=q_1aq_2$ for some contexts $q_1,q_2$ and using
    the first three rules in the definition of $R$ we get that
    $\hole\cdot \alpha(a\hole)\cdot \hole\ R\ \alpha(q_1)\cdot
    \alpha(a\hole)\cdot \alpha(q_2)$ and hence $p\ R\ q$.
  \item If there is a decomposition $p=p_1 p_2$ where $p_1$ and $p_2$ are not
    empty contexts, then from $p\piece q$ there must be a
    decomposition $q=q_1 q_2$ with $p_1 \piece q_1$ and $p_2 \piece
    q_2$.
    By induction we get that $\alpha(p_1)\ R\ \alpha(q_1)$ and $\alpha(p_2)\
    R\ \alpha(q_2)$. Then $\alpha(p)\ R\ \alpha(q)$ follows by using the
    third rule in the definition of $R$.
  \item Suppose now $p=s+\hole$ or $p=\hole+s$. We can assume that $s$ is a tree, since
    otherwise the context $p$ can be decomposed as
    $(s_1+\hole)(s_2+\hole)$. Since $s$ is a tree, it can be decomposed as
    $a(p'0)$, with $a$ being a context with a single letter and the hole below
    and $p'$ a context smaller than $p$. By inspecting the definition of
    $\piece$, there must be some decomposition $q=q_0(a(q'0) + q_1)$ or
    $q=q_0(q_1 + a(q'0))$, with $p' \piece q'$. By the induction assumption,
    $\alpha(p')\ R \ \alpha(q')$.  From this the result follows by applying
    rules three, four and five in the definition of $R$.
  \end{iteMize}
  
  This argument shows that if $v\piece w$ with respect to a particular morphism $\alpha,$ then $v\ R\ w$ and consequently $v\piece w$ with respect to  every morphism.  Thus we have also established the claim made above that the $\piece$ relation on $H$ is independent of the underlying morphism. 
\end{proof}

\section{Piecewise Testable Languages}\label{section-descendant}
The main result in this paper is a characterization of
piecewise testable languages: 

\begin{thm}\label{thm:main} 
  A forest language is piecewise testable if and only if its syntactic algebra
  satisfies the identity
\begin{equation}
  \label{eq:context-absorb}
  u^\omega v = u^\omega = v u^\omega 
\end{equation}
for all $u,v\in V_L$ such that $v\piece u.$
\end{thm}
The identity~(\ref{eq:context-absorb}) is illustrated in
Figure~\ref{fig:rule}.  

\begin{figure}
  \centering
    \includegraphics[scale=0.8]{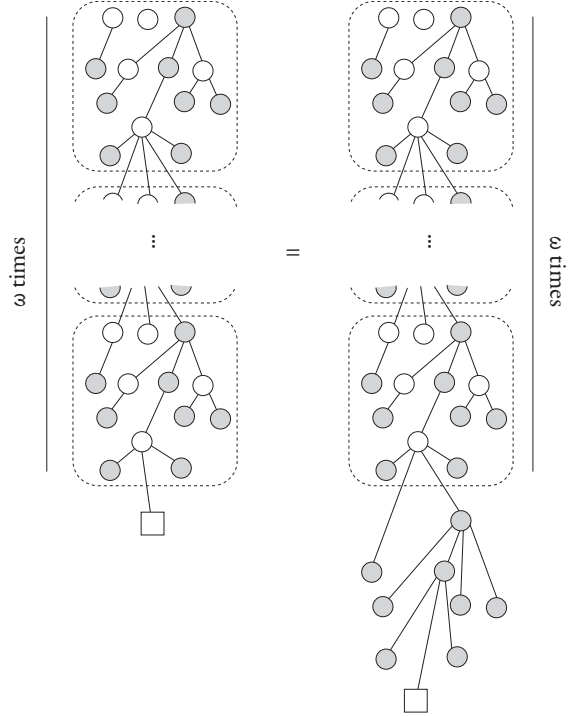}
    \caption{The identity $u^\omega = u^\omega v$, with $v \piece
      u$. The gray nodes are from $v$.}
  \label{fig:rule}
\end{figure}

In view of Corollary~\ref{cor:decidable}, an immediate consequence of Theorem
\ref{thm:main} is that piecewise testability is a decidable property.

\begin{cor}
  It is decidable if a regular forest language is piecewise testable.
\end{cor}
\begin{proof}
  We assume the language is given by its syntactic forest algebra,
  which can be computed in polynomial time from any recognizing forest
  algebra.  The new identities can easily be verified in time polynomial in $|V_L|$ by enumerating all the elements of $V_L$.
\end{proof}

The above procedure gives an exponential upper bound for the complexity in case
the language is represented by a deterministic or even nondeterministic
automaton, since there is an exponential translation from automata into forest
algebras. We do not know if this upper bound is optimal. In contrast, for
languages of words, when the input language is represented by a deterministic
automaton, there is a polynomial-time algorithm for determining piecewise
testability~\cite{stern}.

\medskip

In Sections~\ref{sec:corr-equat} and~\ref{sec:compl-equat}, we prove both
implications of Theorem~\ref{thm:main}. Finally, in
Section~\ref{sec:an-equivalent-set}, we give an equivalent statement of
Theorem~\ref{thm:main}, where the relation $\piece$ is not used.
But before we prove the theorem, we would like to show how it relates to the
characterization of piecewise testable word languages given by Simon.

Let $M$ be a monoid. For $m,n \in M$, we write $m \wordpiece n$ if $m$ is
a---not necessarily connected---subword of $n$, i.e.~there are
elements $n_1,\ldots,n_{2k+1} \in M$ such that
   \begin{eqnarray*}
     n= n_1 \cdots n_{2k} n_{2k+1} \qquad m=n_2 n_4 \cdots n_{2k}\ .
   \end{eqnarray*}
We claim that, using this relation, the word characterization can be
written in a manner identical to Theorem~\ref{thm:main}:
\begin{thm}
  A word language is piecewise testable if and only if its syntactic
  monoid satisfies the identity
\begin{equation}
  \label{eq:monoid-context-absorb}
  n^\omega m = n^\omega = m n^\omega  \qquad \mbox{for }m \wordpiece n\ .
\end{equation}
\end{thm}
\begin{proof}
  Recall that Simon's theorem says a word language is piecewise
  testable if and only if its syntactic monoid is $\mathcal
  J$-trivial. Therefore, we need to show $\mathcal J$-triviality is
  equivalent to~(\ref{eq:monoid-context-absorb}). We use an identity
  known to be equivalent to $\mathcal J$-triviality (see, for instance, ~\cite{eilenbergB}, Sec. V.3.):
\begin{equation}
  \label{eq:j-trivial}
  (nm)^\omega n = (nm)^\omega =  m(nm) ^\omega\ .
\end{equation}  
 Since the above identity is an immediate consequence
of~(\ref{eq:monoid-context-absorb}), it suffices to
derive~(\ref{eq:monoid-context-absorb}) from the above. We only show
$n^\omega m=n^\omega$. As we assume $m \wordpiece n$, there are
decompositions
   \begin{eqnarray*}
     n= n_1 \cdots n_{2k} n_{2k+1} \qquad m=n_2 n_4 \cdots n_{2k}\ .
   \end{eqnarray*}
   By induction on $i$, we show 
   \begin{eqnarray*}
     n^\omega n_i  = n^\omega\ ,
   \end{eqnarray*}
 The result then follows immediately. The base $i=0$, is immediate. In the induction step, we use the
   induction assumption to get:
   \begin{eqnarray*}
     n^\omega n_1 \cdots n_{i-1}  = n^\omega\ .
   \end{eqnarray*}
   By applying~(\ref{eq:j-trivial}), we have 
   \begin{eqnarray*}
    n^\omega =  n^\omega n_1 \cdots n_{i} 
   \end{eqnarray*}
and therefore 
\begin{eqnarray*}
  n^\omega   =    n^\omega    n_i\ .
\end{eqnarray*}
\end{proof}

Note that since the vertical monoid $V$ in a forest algebra is a monoid,
it would make syntactic sense to have the relation $\wordpiece$
instead of $\piece$ in Theorem~\ref{thm:main}. Unfortunately, the
``if'' part of such a statement would be false, as we will show in Section~\ref{sec:an-equivalent-set}. That is why we need to
have a different relation $\piece$ on the vertical monoid, whose
definition involves all parts of a forest algebra, and not just
composition in the vertical monoid.

\subsection{Correctness of the identities}
\label{sec:corr-equat}

In this section we show the easy implication in
Theorem~\ref{thm:main}.
\begin{prop}\label{prop:correcntess}
  If a language is piecewise testable, then its syntactic algebra
  satisfies identity~(\ref{eq:context-absorb}).
\end{prop}

\begin{proof}
  Fix  a language $L$ that is piecewise testable and let $n$ be
  such that membership of $t $  in $L$ only depends on the pieces of $t$
  with at most $n$ nodes.

We will use the following simple fact:
\begin{fact}\label{fact:descendant-obvious}
  If $r$ is any context, $p \piece q$ are contexts and $t$ is a forest, then
  $rpt \piece rqt$.
\end{fact}

We only show the first part of the identity, i.e.
\begin{equation*}
  u^\omega v = u^\omega \hspace{2cm}\text{for $v \piece u$}
\end{equation*}

Fix   $v \piece u$ as above.  By definition of $\omega$, we can
write the identity as an implication: for $k \in \Nat$, if $u^k=u^k
\cdot u^k$ then $u^k \cdot v = u^k$. Let  $k$ be as above. Let $p
\piece q$ be contexts that are mapped to $v$ and $u$ respectively  by the syntactic
morphism of $L$. By unraveling the definition of the syntactic algebra, we need
to show that
\begin{eqnarray*}
    rq^kpt \in L & \qquad  \mbox{ iff } \qquad & rq^kt  \in L
\end{eqnarray*}
holds for any context $r$ and forest $t$.
Consider now the forests 
\begin{eqnarray*}
  rq^{ik}t  \qquad \mbox{and} \qquad rq^{ik}pt \qquad \mbox{for }i \in \Nat \ .
\end{eqnarray*}
As $\hole \piece p \piece q$, thanks to Fact~\ref{fact:descendant-obvious}, we get
\begin{eqnarray*}
    rq^{ik}t\ \piece\ rq^{ik}pt\ \piece\  rq^{(i+1)k}t
\end{eqnarray*}
When $i$ is increasing, the number of pieces of size $n$ of $rq^{ik}t$ is
increasing.  As there are only finitely many pieces of size $n$, for $i$
sufficiently large, the two forests $rq^{ik}t$ and $rq^{(i+1)k}t$ have the same
set of pieces of size $n$. Therefore, for sufficiently large $i$, the two forests 
$rq^{ik}t$ and $rq^{ik}pt$ have the same set of pieces of size $n$, and either both
belong to $L$, or both are outside $L$. However, since
$\alpha_L(q^k)=\alpha_L(q^kq^k),$ we have
\begin{eqnarray*}
  rq^{ik}t\  \in L & \qquad  \mbox{ iff } \qquad & rq^{k}t \in
  L\\
  rq^{ik}pt\  \in L & \qquad  \mbox{ iff } \qquad & rq^{k}pt \in
  L\ ,
\end{eqnarray*}
which gives the desired result.
\end{proof}

\subsection{Completeness of the identities}
\label{sec:compl-equat}
This section is devoted to
showing completeness of the identities: an algebra that satisfies
identity~(\ref{eq:context-absorb}) in Theorem~\ref{thm:main} can only recognize
piecewise testable languages.  We fix an alphabet \A, and a forest language
$L$ over this alphabet, whose syntactic forest algebra $(H_L,V_L)$ satisfies
the identity. We will write $\alpha$ rather than $\alpha_L$ to denote the
syntactic morphism of $L$, and sometimes use the term ``type of $s$'' for the
image $\alpha(s)$ (likewise for contexts).

We write $s \sim_n t$ if the two forests $s,t$ have the same pieces of size
no more than $n$. Likewise for contexts. The completeness part of Theorem~\ref{thm:main}
follows from the following two results.

\begin{lemma}\label{lemma:eilenberg}
  Let $n \in \Nat$. For $k$ sufficiently large, if two forests satisfy
  $s \sim_k s'$, then they have a common piece $t$ in the same
  $\sim_n$-class, i.e.~
  \begin{eqnarray*}
    t \piece s, \quad t \piece s', \quad t \sim_n s, \quad \mbox{and} \quad t \sim_n s'\ .
  \end{eqnarray*}
\end{lemma}

\begin{prop}\label{lemma:remover}
  For $n$ sufficiently large, $pat \sim_n pt$ entails
  $\alpha(pat)=\alpha(pt)$.
\end{prop}

\begin{proof}[Proof of the completeness part of Theorem~\ref{thm:main}]
Take $n$ as in Proposition~\ref{lemma:remover}, and then apply
Lemma~\ref{lemma:eilenberg} to this $n$, yielding $k$.  We show that $s \sim_k
s'$ implies $s \in L \iff s' \in L$, which immediately shows that $L$ is
piecewise testable, by inspecting pieces of size $k$.  Indeed, assume $s \sim_k
s'$, and let $t$ be their common piece as in Lemma~\ref{lemma:eilenberg}. Since
$t$ is a piece of $s$ with the same pieces of size $n$, it can be obtained from
$s$ by a sequence of steps where a single letter is removed in each step without affecting
the $\sim_n$-class.  Each such step preserves the type thanks to
Proposition~\ref{lemma:remover}. Applying the same argument to $s'$, we get
\begin{eqnarray*}
  \alpha(s) = \alpha(t) = \alpha(s')\ ,
\end{eqnarray*}
which gives the desired conclusion.
\end{proof}

We begin by showing Lemma~\ref{lemma:eilenberg}, and then
the rest of this section is devoted to proving
Proposition~\ref{lemma:remover}, the more involved of the two results.

\begin{proof}[Proof of Lemma~\ref{lemma:eilenberg}]
We begin with the following observation.
\begin{fact}\label{fact:pumping}
  Let $n \in \Nat$ and let $K$ be a regular language. There is some constant $k$, such that
  every $t \in K$ contains a piece $s \in K$ of size at most $k$ such that $s
  \sim_n t$.
\end{fact}
\begin{proof}[Proof of Fact~\ref{fact:pumping}]

Let $\beta: \freef \A \to
  (H,V)$ be a morphism into a finite forest algebra.  Let $m=|H|.$  There is a $k$
  such that every forest  $s$ of size greater than $k$ can be written as 
  $s=  q_0q_1 \cdots q_m
  s'$ where $s'$ is a forest and the $q_i$ are nonempty contexts:  this is because every large enough forest
  contains either a collection of $m$ siblings or a chain of length $m.$  It follows that the sequence of
  values $\beta(s'), \beta(q_ms'), \beta(q_{m-1}q_ms'),\ldots, \beta(q_1\cdots q_ms')$ contains
  a repeat, and so we can remove a subsequence of the $q_i$ and obtain a proper piece  $t$ of $s$ such that
  $\beta(s)=\beta(t).$  Thus every forest $s$ has a piece $t$ of size at most $k$ such that $\beta(s)=\beta(t).$
  
  Now let $(H,V)$ be the direct product of the syntactic algebra $(H_K,V_K)$ and the quotient algebra $\freef \A/\sim_n,$ and let $\beta$ be the product of the syntactic moprhism of $K$ and the natural projection onto the quotient by $\sim_n.$  If $s\in K$ then there is a piece $t$ of $s$ of size at most $k$ such that $\beta(s)=\beta(t).$  Thus $t\in K$ and $s\sim_n t,$ proving the Fact.
\end{proof}

We are now ready to prove Lemma~\ref{lemma:eilenberg}. Fix $n\in\Nat$. Notice
that each $\sim_n$ class is a regular language and $\sim_n$ has finitely many
classes. For each $\sim_n$-class $K,$ Fact~\ref{fact:pumping} gives a constant $k_K.$  Let $k$ be the maximum of $n$ and all these $k_K$; we claim the lemma holds for $k.$  Indeed, take
any two forests $s \sim_k s'$. Let $t$ be a piece of $s$ of size at most $k$
with $s \sim_n t$, as given by Fact~\ref{fact:pumping}.  Since $s \sim_k s'$,
the forest $t$ is also a piece of $s'$.  Furthermore since $\sim_k$ implies
$\sim_n$ (by $k \ge n$), we get $s' \sim_n s \sim_n t$, which implies $s'
\sim_n t$ by transitivity of $\sim_n$.
\end{proof}

\medskip

We now show Proposition~\ref{lemma:remover}. Let us fix a context $p$, a label
$a$ and a forest $t$ as in the statement of the proposition.  The context $p$
may be empty, and so may be the forest $t$. We search for the appropriate $n$;
the size of $n$ will be independent of $p,a,t$. We also fix the types
$v=\alpha(p)$, $h=\alpha(t)$ for the rest of this section. In terms of these
types, our goal is to show that $vh=v \alpha(a) h$. To avoid clutter, we will
sometimes identify $a$ with its image $\alpha(a)$, and write $vh=vah$ instead
of $vh=v\alpha(a)h$.

Let $s$ be a forest and $X$ be a set of nodes in $s$. The
\emph{restriction of $s$ to $X$}, denoted $s[X]$, is the piece of $s$
obtained by only keeping the nodes in $X$.

Let $s$ be a forest, $X$ a set of nodes in $s$, and $x \in X$. We say
that $x \in X$ is a $vah$-decomposition of $s$ if: a) if we restrict
$s$ to $X$, remove descendants of $x$, and place the hole in $x$, the
resulting context has type $v$; b) the node $x$ has label $a$; c) if
we restrict $s$ to $X$ and only keep nodes in $X$ that are proper
descendants of $x$, the resulting forest has type $h$.

\begin{df}
  A \emph{fractal} of length $k$ inside a forest $s$ is a sequence
  $x_1 \in X_1\ \cdots\ x_k \in X_k$ of $vah$-decompositions of $s$, where
  $X_i \subseteq X_{i+1} \setminus \set{x_{i+1}}$ holds for $i < k$.
\end{df}

A \emph{subfractal} is extracted by only using a subsequence 
\begin{eqnarray*}
  x_{i_1} \in X_{i_1} \qquad \cdots \qquad x_{i_j} \in X_{i_j}
\end{eqnarray*}
of the
$vah$-decompositions. Such a subsequence  is also a fractal.

\begin{lemma}\label{lemma-fractal}
  Let $k \in \Nat$. For $n$ sufficiently large, $pat \sim_n pt$
  entails the existence of a fractal of length $k$ inside $pat$.
\end{lemma}
\begin{proof}
  The proof is by induction on $k$.  The case $k=1$ is obvious.

  Assume the lemma is proved for $k$ and $n$ and consider the case $k+1$.
	
	The set of forests which have a fractal of length $k$ is a regular language, call it $K$. By Fact~\ref{fact:pumping} applied to $K$, there is some constant $m$  such that every forest in $K$  has a piece that is also in $K$, and whose size is bounded by $m$.  (In this reasoning, we do not use the parameter $n$ of Fact~\ref{fact:pumping}, so we can call Fact~\ref{fact:pumping} with $n=0$). We can assume without loss of generality that $m > n$.   
	 In other words, if  a forest has a fractal of length $k$, then it has a piece of size at most $m$ which has a fractal of length $k$.  This means that if a forest has a fractal of length $k$, then it has a fractal of length $k$ which has at most $m$ nodes (the number of nodes in a fractal is the number of nodes in the largest of its $vah$-decompositions).

  Assume now that $pat \sim_m pt$. By the induction assumption, as $m>n$, we have a
  fractal of length $k$ inside $pat$. From the previous observation, this fractal
  can be assumed to be of size smaller than $m$. Hence we obtain a piece of $pt$
  which is a fractal of length $k$ inside $pt$. Clearly, this resulting fractal can be
  extended to a fractal of length $k+1$ by taking for $X_{k+1}$ all the nodes
  of $pat$ and for $x_{k+1}$ the node $a$.
\end{proof}

Thanks to the above lemma, Proposition~\ref{lemma:remover} is a consequence of
the following result:
\begin{prop}
  For $k$ sufficiently large, the existence of a fractal of length $k$
  inside $pat$ entails $vh=vah$.
\end{prop}

The rest of this section is devoted to a proof of this proposition.
The general idea is as follows. Using some simple combinatorial
arguments, and also Ramsey's Theorem, we will show that there is
also a large subfractal whose structure is very regular, or tame, as
we call it. We will then apply identity~(\ref{eq:context-absorb}) to
this regular fractal, and show that a node with label $a$ can be
eliminated without affecting the type.

\begin{figure}
  \centering
  \includegraphics[scale=0.8]{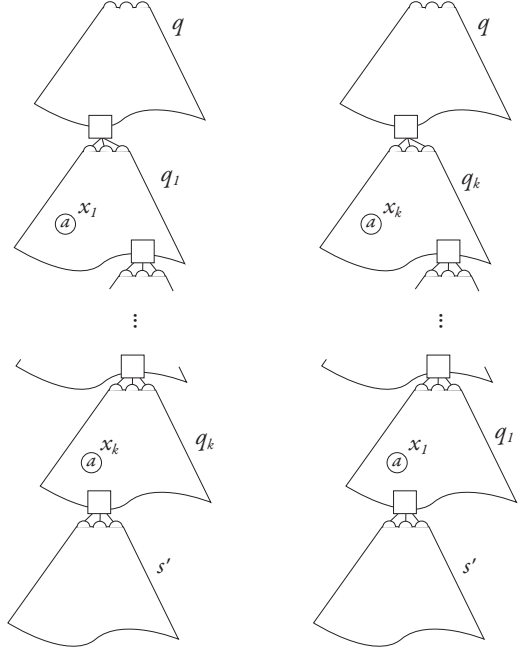}
  \caption{Two types of tame fractal.}
\label{fig:tame}
\end{figure}

A fractal $x_1 \in X_1\ \cdots\ x_k \in X_k$ inside a forest $s$ is
called \emph{tame} if $s$ can be decomposed as $s=q q_1 \cdots q_ks'$
(or $s=q q_k \cdots q_1s'$) such that for each $i=1,\ldots,k$, the
node $x_i$ is part of the context $q_i$, see Fig.~\ref{fig:tame}.
This does not necessarily mean that the nodes $x_1,\ldots,x_k$ form a chain,
since some of the contexts $q_i$ may be of the form $\hole + t$.

\begin{lemma}\label{lemma:decompose}
  Let $k \in \Nat$. For $n$ sufficiently large, if there is a fractal
  of length $n$ inside $pat$, then there is a tame fractal of length $k$ inside
  $pat$.
\end{lemma}
\begin{proof}
The main step is the following claim.

\begin{claim}\label{claim:decompose}
  Let $m \in \Nat$. For $n$ sufficiently large, for every forest $s$, and every
  set $X$ of at least $n$ nodes, there is a decomposition $s=q q_1 \cdots q_m
  s'$ where every context $q_i$ contains at least one node from $X$.
\end{claim}
\proof
  Let $Y$ be the smallest set of nodes that contains $X$ and is closed under
  closest common ancestors. If $n$ is chosen large enough, either $s[Y]$
  consist of more than $m$ trees, or it contains a node having more than $m$
  children, or $s[Y]$ contains a chain of length bigger than $m$. We are thus
  left with three cases:
\begin{iteMize}{$\bullet$}
\item In the set $Y$, there is a path $y_1 < \cdots < y_{m+1}$.  For $i \in
  \set{1,\ldots,m+1}$, consider the set of nodes
	\begin{align*}
		Y_i = \set{ z : z \ge y_i \mbox{ and }z \not \ge y_{i+1}}.
	\end{align*}
	Each set $Y_i$ contains at least one node of $X$, by definition of the
        set $Y$.  The decomposition in the statement of the lemma is chosen so
        that context $q_i$ corresponds to the set $Y_i$.  The context $q$
        corresponds to all nodes that are not descendants of $y_1$, and the
        forest $s'$ corresponds to all descendants of $y_{m+1}$.

\item There is a node $y \in Y$ such that at least $m+1$ children of
        $y$ have some node from $Y$ (and therefore also $X$) in their subtree.
        Let $t$ be the forest containing all proper descendants of $y$.  By
        assumption on $y$, the forest $t$ can be decomposed as $t=t_1 + \cdots
        + t_{m+1}$ so that each of the forests contains at least one node from $X$.
        For the decomposition in the statement of the lemma, we define $q$ to
        be the set of nodes outside $t$, which includes $y$, and we define
        $q_i$ to be $t_i + \hole$ and $s'$ as $t_{m+1}$.
\item The forest $s$ can be decomposed as $t=t_1 + \cdots + t_{m+1}$ so that
  each of the forests contains at least one node from $X$. We conclude as in the
  previous case but with an empty $q$.\qed
\end{iteMize}

\noindent We now come back to the proof of the lemma.  For $k\in\Nat$ let $n$ be
the number defined by Claim~\ref{claim:decompose} for $m=k^2$. Let
$x_1 \in X_1\ \cdots\ x_n \in X_n$ be a fractal of length $n$ inside $s=pat$. We
apply Claim~\ref{claim:decompose}, with $X=\set{x_1,\ldots,x_n}$ and
obtain a decomposition $s=q q_1 \cdots q_m s'$. For each
$i=1,\ldots,m$ the context $q_i$ contains at least one node of $X$.
We chose arbitrarily one of them and denote it by
$x_{n_i}$. Unfortunately, the function $i \mapsto n_i$ need not be
monotone, as required in a tame fractal. However, we can always
extract a monotone subsequence, since any number sequence of length
$k^2$ is known to have a monotone subsequence of
length~$k$~\cite{erdosszekeres}
\end{proof}

We now assume there is a tame fractal $x_1 \in X_1\ \cdots\ x_k \in
X_k$ inside $s=pat$, which is decomposed as $s=qq_1\cdots q_k s'$,
with the node $x_i$ belonging to the context $q_i$. The dual case when
the decomposition is $s=qq_k \cdots q_1 s'$, corresponding to a
decreasing sequence in the proof of Lemma~\ref{lemma:decompose}, is
treated analogously.

The general idea is as follows. We will define a notion of
monochromatic tame fractal, and show that $vah=vh$ follows from the
existence of large enough monochromatic tame fractal. Furthermore, a
large monochromatic tame fractal can be extracted from any
sufficiently large tame fractal thanks to the Ramsey Theorem.

Let $i,j,l$ be such that $0\leq i < j \leq l \leq k$. We define
$u_{ijl}$ to be the image under $\alpha$ of the context obtained from
$q_{i+1} \cdots q_{j}$ by only keeping the nodes from $X_l$ (with the
hole staying where it is). We define $w_{ijl}$ to be the image under
$\alpha$ of the context obtained from $q_{i+1} \cdots q_j$ by only
keeping the nodes from $X_l \setminus \set {x_l}$.  Straight from this
definition, as $X_l \subseteq X_{l+1}$ we have
\begin{equation}
  \label{eq:vert-piece1}
    w_{ijl}  \piece u_{ijl} \text{ and } u_{ijl}  \piece u_{ij(l+1)}
\end{equation}

A tame fractal is called \emph{monochromatic} if for all $i < j < l$ and all
$i' < j' < l'$ taken from $\set{1,\ldots,k}$, we have
\begin{eqnarray*}
  u_{ijl}= u_{i'j'l'}\ .
\end{eqnarray*}
Note that in the above definition, we require $j < l$, even though
$u_{ijl}$ is defined even when $j \le l$.

We apply the following form of Ramsey's Theorem (see, for example, Bollobas~\cite{bollobas}):  Let $c, r, k$ be positive integers. Then there exists an integer $N$ with the following property.  Let $|S|\geq N,$ and suppose that the subsets of $S$ of cardinaility $r$ are colored with $c$ colors. Then there exists a subset $T$ of $S$ with $|T|\geq k$ such that all subsets of $T$ with of cardinality $r$ have the same color.

Let $\omega$ be the exponent associated to the syntactic forest algebra
$(H_L,V_L)$ as defined in Section~\ref{sec:forest-algebras}. If there is a tame fractal of size $N$ inside $s,$ then the map $\{i,j,l\}\mapsto u_{ijl}$ gives us a coloring of the cardinality 3 subsets of $\{1,\ldots,N\}$ with $|V_L|$ colors.  By Ramsey's Theorem, if $N$ is sufficiently large, there is a  monochromatic fractal of length~$k=\omega+1$ inside $s$.

We conclude by showing the following result:
\begin{lemma}\label{lemma:mono-vertical}
  If there is a monochromatic tame fractal of length~$k=\omega+1$ inside
  $pat=qq_1\cdots q_k s'$, then $vah=vh$.
\end{lemma}
\begin{proof}
 Fix a monochromatic tame fractal $x_1 \in X_1\
  \cdots\ x_k \in X_k$ inside a forest $s=pat=qq_1\cdots q_ks'$.  Since
  $x_k \in X_k$ is a $vah$-decomposition, the statement of the lemma
  follows if $\alpha$ assigns the same type to the two restrictions
  $s[X_k]$ and $s[X_k \setminus \set{x_k}]$.

  Recall the definition of $u_{ijl}$ and $w_{ijl}$ above.  The type of
  the forest $s[X_k]$ can be decomposed as
\begin{eqnarray*}
  \alpha(s[X_k]) = \alpha(q[X_k]) \cdot 
  u_{01k} \cdot u_{12k} \cdot u_{23k} \cdots u_{(k-1)kk} \cdot \alpha(s'[X_k])
\end{eqnarray*}
The type of $s[X_k \setminus \set{x_k}]$ is decomposed the same way,
only $u_{(k-1)kk}$ is replaced by $w_{(k-1)kk}$.  Therefore, the lemma
will follow if
\begin{eqnarray*}
  u_{01k}\cdot u_{12k} \cdot u_{23k} \cdots u_{(k-1)kk} = 
  u_{01k} \cdot u_{12k} \cdot u_{23k} \cdots w_{(k-1)kk}\ .
\end{eqnarray*}
Since the fractal is monochromatic, and since $k=\omega+1$ the above becomes
\begin{eqnarray*}
  u_{01k}^\omega \cdot u_{(k-1)kk} =   u_{01k}^\omega \cdot w_{(k-1)kk} \ .
\end{eqnarray*}
By~(\ref{eq:vert-piece1}) and monochromaticity we have 
\begin{eqnarray*}
  w_{(k-1)kk} & \piece& u_{(k-1)k(k+1)}= u_{01k} \\
u_{(k-1)kk} & \piece & u_{(k-1)k(k+1)}=u_{01k}\ .
\end{eqnarray*}
Therefore identity~(\ref{eq:context-absorb}) can be applied to show
that both sides are equal to $u_{01k}^\omega$.  Note that we use only
one side of identity~(\ref{eq:context-absorb}), $u^\omega v=u^\omega$
. We would have used the other side when considering the case when
$s=qq_k\cdots q_1s'$.
\end{proof}

\subsection{An equivalent set  of identities}
\label{sec:an-equivalent-set}

\begin{figure*}[htbp]
\centering
  \includegraphics[scale=0.8]{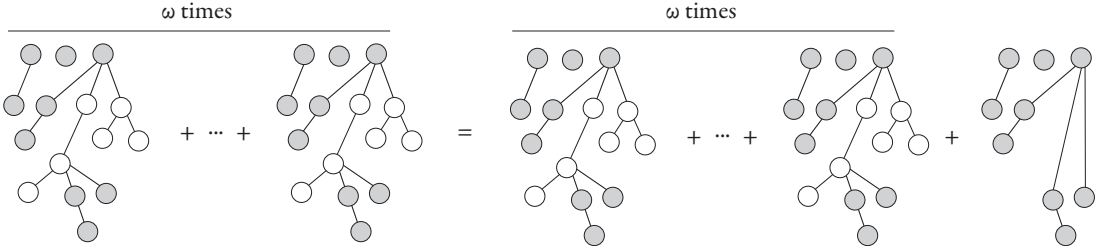}
  \caption{The identity     $\omega(vuh)=\omega(vuh) + vh$, with the white nodes belonging to
    $u$.}
  \label{fig:omega}
\end{figure*}

In this section, we rephrase the identities used in
Theorem~\ref{thm:main}.  There are two reasons to rephrase the identities.

The first reason is that identity~(\ref{eq:context-absorb}) refers to
the relation $v \piece w$. One consequence is that we need to prove
Corollary~\ref{cor:decidable} before concluding that
identity~(\ref{eq:context-absorb}) can be checked effectively.

The second reason is that we want to pinpoint how
identity~(\ref{eq:context-absorb}) diverges from $\mathcal J$-triviality of the
context monoid $V$. Consider the forest language ``all trees in the forest are
of the form $aa$''. It is easy to verify that the syntactic forest algebra of
this language is such that $V$ is $\mathcal J$-trivial. But this language is
not piecewise testable, since for any $k>0,$ the forests $k\cdot aa$ and
$k\cdot aa +a$ contain the same pieces of size at most $k,$ but the first of these
forests is in the language, while the second is not.

The proposition below identifies an additional
condition (depicted in Figure~\ref{fig:omega}) that must be added to $\mathcal J$-triviality.

\begin{prop}\label{prop:other-eq}
  Identity~(\ref{eq:context-absorb}) is equivalent to $\mathcal
  J$-triviality of $V$, and the identity
\begin{equation}\label{eq:new-eq}
  vh+  \omega\cdot vuh  =   \omega \cdot vuh = \omega \cdot vuh + vh
\end{equation}  
\end{prop}
\begin{proof}
  One implication is obvious: both $\mathcal J$-triviality
  and~(\ref{eq:new-eq}) follow from~(\ref{eq:context-absorb}). For the other
  implication, we assume $V $ is $\mathcal J$-trivial and that
  ~(\ref{eq:new-eq}) holds.  We must show that if $v \piece u$, then
\begin{eqnarray*}
  u^\omega v = u^\omega = v u^\omega \ .
\end{eqnarray*}
We will only show the first equality, the other is done the same way.
By unraveling the definition of $v \piece u$, there is a morphism
\begin{eqnarray*}
  \alpha : \freef \A \to (H,V)
\end{eqnarray*}
and two contexts $p \piece q$ over \A such that $\alpha(p)=v$ and
$\alpha(q)=u$.

The proof goes by induction on the size of $p$.

If $p$ can be decomposed as $p_1 p_2$ with $p_1,p_2$ nonempty, then we have
$p_1 \piece q$ and $p_2 \piece q$ and, by induction, $\alpha(q)^\omega \cdot
\alpha(p_1)=\alpha(q)^\omega$, $\alpha(q)^\omega \cdot
\alpha(p_2)=\alpha(q)^\omega$. Hence we get:
\begin{eqnarray*}
  \alpha(q)^\omega \cdot \alpha(p_1) \cdot \alpha(p_2) = \alpha(q)^\omega \cdot
  \alpha(p_2) = \alpha(q)^\omega\ .
\end{eqnarray*}

If $p$ consists of single node with a hole below, then we have $q=q_0
p q_1$ for some two contexts $q_0,q_1$, and therefore also $u=u_0 v
u_1$ for some $u_0,u_1$. The result then follows by $\mathcal J$-triviality of
$V$ (recall that $\mathcal J$-triviality implies identity~\eqref{eq:j-trivial}):
\begin{eqnarray*}
  u^\omega v = (u_0vu_1)^\omega v = (u_0vu_1)^\omega u_0 v = (u_0vu_1)^\omega = u^\omega \ .
\end{eqnarray*}
In the above, we used twice identity~\eqref{eq:j-trivial}: Once when adding $u_0$ to $u^\omega$, and then when
removing $u_0v$ from after $u^\omega$.

The interesting case is when $p=\hole + s$ for some tree $s$. In this
case, the context $q$ can be decomposed as $q_1(\hole+t)q_2$, with $s
\piece t$. We have
\begin{eqnarray*}
  u^\omega v = \alpha(q_1(\hole+t)q_2)^\omega \alpha(\hole+s)\ .
\end{eqnarray*}
 Thanks to identity~\eqref{eq:j-trivial}, the above can be rewritten as
 \begin{eqnarray*}
  u^\omega v =  \alpha(q_1(\hole+t)q_2)^\omega (\alpha(\hole+t))^\omega
  \alpha(\hole+s)\ .
 \end{eqnarray*}
 Notice now that 
\begin{eqnarray*}
(\alpha(\hole+t))^\omega \alpha(\hole+s) = ( \hole + \alpha(s) +  \omega \cdot
\alpha(t)  )\ .
 \end{eqnarray*}
It is therefore sufficient to show that $s \piece t$ implies
\begin{eqnarray*}
  \omega \cdot \alpha(t) =\alpha( s) +  \omega \cdot \alpha(t)\ .
\end{eqnarray*}
The proof of the above equality is by induction on the number of nodes
that need to be removed from $t$ to get $s$.  The base case $s=t$
follows by aperiodicity of $H$, which follows by aperiodicity of $V$,
itself a consequence of $\mathcal J$-triviality.  Consider now the case when
$t$ is bigger than $s$. In particular, we can remove a node from $t$
and still have $s$ as a piece. In other words, there is a
decomposition $t=q_0q_1 t'$ such that $s \piece q_0 t'$. Applying the
induction assumption, we get
\begin{eqnarray*}
  \omega \cdot \alpha(q_0 t') = \alpha(s) +  \omega \cdot \alpha( q_0 t') \ .
\end{eqnarray*}
Furthermore, applying identity~(\ref{eq:new-eq}), we get
\begin{eqnarray*}
  \omega \cdot \alpha(t) = \alpha(q_0 t') +  \omega \cdot \alpha(t)  = 
  \omega \cdot \alpha(q_0 t') + \omega \cdot \alpha(t)\ \ .
\end{eqnarray*}
Combining the two equalities, we get the desired result.
\end{proof}

\newcommand{\ccapiece}{\unlhd}
\section{Closest common ancestor}\label{section-cca}

According to the definition of piece in Section~\ref{sec:notation}, $t=d(a+b)$
is a piece of the forest $s=dc(a+b).$ In this section we consider a notion of
piece which does not allow removing the closest common ancestor of two nodes,
in particular removing the node $c$ in the example above. The logical
counterpart of this notion is a signature where the closest common ancestor
(a three argument predicate) is added.

Recall that in a forest $s$ we say that a node $z$ is the \emph{\lca} of the
nodes $x$ and $y$, denoted $z=x \sqcap y$, if $z$ is an ancestor of both $x$
and $y$ and all other nodes of $s$ with this property are ancestors of
$z$. Note that the ancestor relation can be defined in terms of the \lca, since
a node $x$ is an ancestor of $y$ if and only if $x$ is the \lca of $x$ and
$y$. We now say that a forest $s$ \emph{is a \lcapiece} of a forest $t$, and
write this as $s \ccapiece t$, if there is an injective mapping from nodes of
$s$ to nodes of $t$ that preserves the label of the node together with the forest-order and the \lca relationship
(the ancestor relationship is then necessarily preserved).  An equivalent
definition is that the \lcapiece relation is the reflexive transitive closure
of the relation
\begin{eqnarray*}
\set{  (pt, pat) : \mbox{$p$ is a context, $a$ is a node, $t$ is a \emph{tree} or empty}}
\end{eqnarray*}
Notice the difference with the notion of piece as defined in
Section~\ref{sec:notation}, where $t$ could be an arbitrary forest.
Similarly we say that a context $p$ is a \lcapiece of the context $q$, $p
\ccapiece q$, if there is an injective mapping from $p$ to $q$ as above that also preserves the hole.

A forest language $L$ is called \emph{\lcapiecewise testable} if there exists
$n >0$ such that membership of $t$ in $L$ depends only on the set of
{\lcapiece}s of $t$ of size $n$.

As before, every \lcapiecewise testable language is regular and an analogue of
Proposition~\ref{boolean_sigma_1} holds as well.

\begin{prop}\label{boolean_sigma_1-cca}
 A forest language is \lcapiecewise testable 
  iff it is definable by a Boolean combination of $\Sigma_1(\sqcap,\orderfo)$ formulas.
\end{prop}

Recall that the ancestor relation can be expressed using the \lca relation
hence $\Sigma_1(\sqcap,\orderfo)$ could be replaced by
$\Sigma_1(\sqcap,\orderfo,<)$ in the statement of
Proposition~\ref{boolean_sigma_1-cca}.
A first remark is that there are more \lcapiecewise testable languages
than there are piecewise testable ones. Hence the identities that
characterize piecewise testable languages are no longer valid. In
particular, in the syntactic algebra of a \lcapiecewise testable
language, the context monoid $V$ may no longer be $\mathcal J$-trivial. To
see this consider the language $L$ of forests over $\{a,b,c\}$ that
contain the \lcapiece $a(b+c)$. This is the language ``some $a$ is the
\lca of some $b$ and $c$''.  Then, for all $n$, the context $p=(ab)^n\hole$ is
not the same as the context $q=(ab)^na\hole$ as $p(b+c) \not\in L$
while $q(b+c) \in L$. Hence the identity $(uv)^\omega=(uv)^\omega u$ does not
hold in the syntactic context monoid of $L$. However as we noted earlier, any $\mathcal
J$-trivial monoid satisfies this identity.
Note however that $p$ and $q$ satisfy the equivalence $pt \in L$ iff $qt \in L$
for all {\em trees} $t$.  The characterization below is a generalization of
this idea of distinguishing trees from forests.

We call a context a \emph{tree-context} if it is nonempty and has one node that
is the ancestor of all other nodes, including the hole.

In the presence of the \lca,  the algebraic situation is more complicated as well: \lcapiecewise testability of a forest language $L$ is not determined by the
syntactic forest algebra alone. To obtain an algebraic characterization of this class of languages,  it is necessary to look at the
\emph{syntactic morphism} $\alpha_L:\A^{\Delta}\to (H_L,V_L)$ that maps each
$(h,v)$ to its $\sim_L$-class, and not just the  the image of
this morphism. (We can be considerably more precise about this:  The distinction is that the \lcapiecewise testable languages
do not form a variety of languages in the sense described by Eilenberg~\cite{eilenbergB}.  In particular, this family of languages lacks
the crucial property of being closed under inverse images of morphisms between free forest algebras; this fails if the morphism maps some generator $a\hole$ to the empty
context, or to a context of the form $p+s,$ where $p$ is a context and $s$ is a nonempty forest.  However
\lcapiecewise testable languages satisfy all the other properties of varieties
of languages and in particular they are closed under  inverse images of
homomorphisms that are ``tree-preserving'', i.e., the image of $a\hole$ is a
tree-context $p$ for all $a$. Varieties of forest languages are discussed in~\cite{varieties-forest-algebra}.)

We extend the \lcapiece relation to elements of a forest algebra $(H,V)$ in the
presence of a morphism $\alpha:\A^{\Delta}\to (H,V)$ as follows: we write
$v \ccapiece w$ if there are contexts $p \ccapiece q$ that are mapped to $v$
and $w$ respectively by the morphism $\alpha$. There is a subtle difference
here with the definition of $\piece$ defined in Section~\ref{sec:notation}: the
$\ccapiece$ relation on $V$ depends on the morphism $\alpha$! Similarly we
define the notion of $g\ccapiece h$ for $g,h \in H$.

 The elements of $V$ that are images under the morphism
$\alpha$ of a tree-context are called tree-context-types.  Similarly, the
elements of $H$ that are images of a tree are called tree-types (it is possible
for an element to be an image of both a tree and a non-tree, but it is still
called a tree-type here). Note that the notions of tree-type and of
tree-context-type are relative to $\alpha$.

\begin{thm}\label{thm:main-lca}
  A forest language $L$ is \lcapiecewise testable if and only if its
  syntactic algebra and syntactic morphism satisfy the following identities:
\begin{equation}\label{eq:lca-absorb}
 u^\omega h = u^\omega v h = v u^\omega h 
\end{equation}
whenever $h$ is a tree-type or empty, and $v \ccapiece u$ are tree-context-types, and
\vspace{-.1cm}
\begin{equation}\label{eq:lca-forest-absorb}
 \omega \cdot h = \omega \cdot h + g =g+ \omega \cdot h\hspace{2cm}\text{if $g \ccapiece h$}
\end{equation}
\end{thm}

Because of the finiteness of the syntactic forest algebra $(H_L,V_L)$ one can
effectively decide whether an element of one of these monoids is the image of a
tree-context or of a tree.  Whether or not $v \ccapiece u$ or $g \ccapiece h$
holds can be decided in polynomial time using an algorithm as in
Corollary~\ref{cor:decidable} based on the following equivalent definition of
$\ccapiece$: Let $(H,V)$ be a forest algebra and $\alpha$ a surjective morphism
from $\A^{\Delta} \to (H,V)$. Let then $R$ be the
smallest relation on $V$ that satisfies the following rules, for all $v,v',w,w'
\in V$:
\begin{eqnarray*}
  \begin{array}{rcll}
      \hole & R &  v\\
      \alpha(a)v & R & \alpha(a)v' \qquad & \mbox{ if $v\ R\ v'$}\\
  vw & R & v'w' \qquad & \mbox{ if $v\ R\ v'$ and $w\ R\ w'$ and $w,w'$ are
    tree-context-types}\\
  vw & R & v'w' \qquad & \mbox{ if $v\ R\ v'$ and $w\ R\ w'$ and $v,v'$
    are of the form $(s+\hole+t)$}\\
  \hole + v0 & R & \hole +v'0 & \mbox{ if $v\ R\ v'$}\\
   v0 + \hole & R & v'0 + \hole  & \mbox{ if $v\ R\ v'$}
  \end{array}
\end{eqnarray*}

\begin{lemma}\label{lemma:equivalent-def-cca}
  For any finite $(H,V)$ and surjective morphism $\alpha$, the relations $R$ and
  $\ccapiece$ are the same.
\end{lemma}
\proof
  We first show the inclusion of $R$ in $\ccapiece$. A simple induction on the
  number of steps used to derive $v \ R \ w$, produces contexts $p \ccapiece q$
  with $\alpha(p)=v$ and $\alpha(q)=w$. Moreover $p$ ($q$) is a
  tree-context whenever $u$ ($v$) is a tree-context-type. The surjectivity of
  $\alpha$ is necessary for starting the induction in the case $\hole\ R\ v$.

  For the inclusion of $\ccapiece$ in $R$, we show that $\alpha(p)\ R \
  \alpha(q)$ holds for all contexts $p \ccapiece q$. The proof is by
  induction on the size of $p$:
  \begin{iteMize}{$\bullet$}
  \item If $p$ is the empty context, then the result follows thanks to
    the first rule in the definition of $R$.  If $p=a\hole$ then from $p
    \ccapiece q$ it follows that $q=q_1aq_2$ for some contexts $q_1,q_2$ and using
    the first and second rule in the definition of $R$ we get that
    $\hole\ R\ \alpha(q_1)$, $\hole\ R\ \alpha(q_2)$, and $\alpha(a) R
    \alpha(a)\alpha(q_2)$. Hence using the third rule in the definition of $R$
    we get the desired result by composition.
  \item If there is a decomposition $p=p_1 a p_2$ where $p_1,p_2$ are contexts,
    then from $p\ccapiece q$ there must be a decomposition $q=q_1 a q_2$ with
    $p_1 \ccapiece q_1$ and $p_2 \ccapiece q_2$.
    By induction we get that $\alpha(p_1)\ R\ \alpha(q_1)$ and $\alpha(p_2)\
    R\ \alpha(q_2)$. Applying the second rule to the latter we get that $\alpha(ap_2)\
    R\ \alpha(aq_2)$. We can now apply the third rule to derive $\alpha(p)\ R\
    \alpha(q)$. 
  \item If there is a decomposition $p=p_1p_2$ where $p_1,p_2$ are non empty
    contexts and $p_1$ is of the form $(s+\hole + t)$, then from $p\ccapiece q$
    there must be a decomposition $q=q_1 q_2$ with $p_1 \ccapiece q_1$ and $p_2
    \ccapiece q_2$ and where $q_1$ is of the form $(s'+\hole+t')$. We conclude
    by induction and using the fourth rule in the definition of $R$.
  \item The remaining case is when $p=(t+\hole)$ (or $p=\hole+t$) where $t$ is
    a tree of the form $ap'0$ for some context $p'$. Then from $p\ccapiece q$
    we have $q=aq'0+q_1$ for some contexts $q_1,q'$, with $p' \ccapiece
    q'$. By induction we have $\alpha(p')\ R\ \alpha(q')$. Using the second
    rule we get $\alpha(ap')\ R\ \alpha(aq')$. Using the last rule
    we get $\alpha(p)\ R\ \alpha(aq'0+\hole)$. By the first rule we
    have $\hole\ R\ \alpha(q_1)$. We conclude using the
    fourth rule.\qed
\end{iteMize}

\noindent This implies that Theorem~\ref{thm:main-lca} yields a decidable
characterization  of the \lcapiecewise testable languages.

\begin{cor}
  It is decidable if a regular forest language is \lcapiecewise testable.
\end{cor}

The proof of Theorem~\ref{thm:main-lca} follows the same outline as that of the proof of
Theorem~\ref{thm:main}, but the details are somewhat complicated. 

\subsection{Proof of Theorem~\ref{thm:main-lca}}

The proof that~(\ref{eq:lca-absorb}) and~(\ref{eq:lca-forest-absorb})
are necessary is the same as Section~\ref{sec:corr-equat}. The only
difference is that instead of Fact~\ref{fact:descendant-obvious}, we
use the following.
\begin{fact}\label{fact:lca-descendant-obvious}
  If $r$ is any context, $p \ccapiece q$ are tree-contexts, and $t$ is a
  tree or empty, then $rpt \ccapiece rqt$.
\end{fact}

We now turn to the completeness proof in Theorem~\ref{thm:main-lca}.
The proof is very similar to the one of the previous section, with
some subtle differences.

As before, we fix a language $L$ whose syntactic forest tree algebra
$(H,V)$ satisfies all the identities of
Theorem~\ref{thm:main-lca}. We write $\alpha$ for the syntactic
morphism.

We now write $s \sim_n t$ if the two forests $s,t$ have the same \lcapieces of
size $n$. Likewise for contexts.

The main step is to show the following proposition.

\begin{prop}\label{prop:remover-lca}
 For $n$ sufficiently large, if $t$ is a tree or empty, then $pat
 \sim_n pt$ entails $\alpha(pat)=\alpha(pt)$.
\end{prop}

Theorem~\ref{thm:main-lca} follows from the above proposition in the same way
as Theorem~\ref{thm:main} follows from Proposition~\ref{lemma:remover} in the
previous section. The reason why we assume that $t$ is either a tree or empty
is because when $s$ is an \lcapiece of $s'$, then $s$ can be obtained from $s'$
by iterating one of the following two operations: removing a leaf, or removing
a node which has only one child. Hence during the pumping argument yielding
Theorem~\ref{thm:main-lca} from Proposition~\ref{prop:remover-lca} it is enough
to preserve the type only for these operations. We thus concentrate on showing
Proposition~\ref{prop:remover-lca}.

We will now redefine the concept of fractal for our new, \lca setting.  The key
change is in the concept of a $vah$-decomposition. We change the notion of $x
\in X$ being a $vah$-decomposition of $s$ as follows: all conditions of the old
definition hold, but new conditions are added. First we require that $s[X]$ be
a \lca piece of $s$, in particular this implies that if two elements of $X$
have a \lca in $s$ then this \lca is also in $X$. Moreover either $x$ has no descendants in
$X$; or there is a minimal element of $X$ that has $x$ as a proper ancestor. In
other words, the part of $s[X]$ that corresponds to $h$ is either empty, or is
a tree.  In particular, $s[X \setminus \set x]$ is a \lca piece of $s[X]$;
which is the key property required below. From now on, when referring to a
$vah$-decomposition, we use the new definition. In particular in the concept of
a fractal $x_1 \in X_1,\ldots,x_k \in X_k$ inside $s$ we now have that for each
$i$, $x_i \in X_i$ is a $vah$-decomposition of $s$ in the new sense.

The proof of the following lemma is exactly the same as its counterpart in
Section~\ref{sec:compl-equat} (Lemma~\ref{lemma-fractal}) and is therefore omitted.
\begin{lemma}
 Let $k \in \Nat$. For $n$ sufficiently large, if $t$ is a tree or
 empty, then $pat \sim_n pt$ entails the existence of a fractal of
 length $k$ inside $pat$.
\end{lemma}

A fractal $x_1 \in X_1\ \cdots\ ,x_k \in X_k$ inside $s$ is called
\emph{\lcatame} if $s$ can be decomposed as $s=q q_1 \cdots q_ks'$ (or
$s=q q_k \cdots q_1s'$) such that $x_1 \in q_1, \cdots, x_k \in q_k$
and such that either:
\begin{iteMize}{$\bullet$}
\item Each $q_i$ is a tree context whose root node belongs to $X_i
  \setminus \set{x_i}$.
\item Each $q_i$ is a context of the form $\hole + t_i$, with $t_i$ a
  forest.
\end{iteMize}

\begin{lemma}\label{lemma:lca-decompose}
 Let $k \in \Nat$. For $n$ sufficiently large, if there is a fractal
 of length $n$ inside $pat$, then there is a \lcatame fractal of length $k$
 inside $pat$.
\end{lemma}
\proof
  The proof is essentially the same as for the counter part in
  Section~\ref{sec:compl-equat} (Lemma~\ref{lemma:decompose});
  only this time we need to be more careful to satisfy the more
  stringent requirements in a \lcatame fractal.

  Let $m=2k+2$. Using the same reasoning as in the proof of
  Lemma~\ref{lemma:decompose}, if $n$ is large enough then we may extract a
  subfractal of length $m$ where either:
 \begin{iteMize}{$\bullet$}
\item All the nodes $x_1,\ldots,x_m$ have the same \lca. In this case,
     we can extract a \lcatame subfractal, where each context is of
     the form $\hole + t_i$.
\item The set $Y=\set{y : y \mbox{ is a \lca of some } x_i,x_j}$ contains a
     chain $y_1 < \cdots < y_{m}$, such that for each $i \leq m$, the set 
$Y_i = \set{ z : z \ge y_i \mbox{ and }z \not \ge y_{i+1}}$ contains at least one of the node $x_i$. (There is a second case,
     where the nodes $y_1,\ldots,y_m$ are ordered the other way: with
     $y_{i+1}$ an ancestor of $y_i$. This case is treated
     analogously.) In particular, $y_i$ is the \lca of $x_i$ and any
     of the nodes $x_{i+1},\ldots,x_m$. Since $X_{i+1}$ contains both
     $x_i$ and $x_{i+1}$, each node $y_i$ belongs to the set
     $X_{i+1}$. As we may have $x_i=y_i$, the desired \lcatame fractal is
     obtained as follows: We use
     $x_2 \in X_2, x_4 \in X_4, \ldots, x_{2k} \in X_{2k}$ as the
     fractal (recall that $m=2k+2$); while the decomposition $q q_1
     \ldots q_k s'$ is chosen so that $q_i$ has its root in
     $y_{2i-1}$, and its hole in $y_{2i+1}$.\qed
\end{iteMize}

\noindent Recall the definition of $u_{ijl}$ and $w_{ijl}$ as the image under
$\alpha$ of the context obtained from $q_{i+1} \cdots q_j$ by
restricting $s$ to $X_l$ and $X_l \setminus \set {x_l}$, respectively.
Note that because of the new definition of fractals we have:
\begin{equation}
 \label{eq:lca-prop-decompose1}
   w_{ijl}  \ccapiece u_{ijl}  \qquad \text{and}\qquad    u_{ijl}  \ccapiece u_{ij(l+1)} 
\end{equation}
\begin{equation}
 \label{eq:lca-prop-decompose2}
 \text{if the $q_i$ are tree-contexts then } u_{ijl}, w_{ijl} \text{ are tree-context-types}
\end{equation}

The definition of monochromaticity is the same as in the previous section and
Ramsey's Theorem gives.

\begin{lemma}
 If there is a \lcatame fractal of sufficiently large size inside $pat$, then there is a
 monochromatic \lcatame fractal of size $m=\omega+2$ inside $pat$.
\end{lemma}

We will now take a monochromatic \lcatame fractal, and conclude by showing that
$\alpha(pat)=\alpha(pt)$. 

\begin{lemma}\label{lemma:mono-vertical-gca}
 If there is a monochromatic \lcatame fractal of size $\omega+2$ inside $pat$, then
 $vah=vh$.
\end{lemma}
\begin{proof}
 Fix a monochromatic \lcatame fractal of size $m= \omega +2$ and let $k=m-1$.
 Since $x_k \in X_k$ is a $vah$-decomposition, the statement of the
 lemma follows once we show that $\alpha$ assigns the same type to
 the forest $s[X_k]$ and $s[X_k \setminus \set{x_k}]$.

 Recall that the type of the forest $s[X_k]$ can be decomposed as follows (the case
 where $s=qq_mq_{m-1}\cdots q_1s'$ is treated similarly by duality).
\begin{eqnarray*}
 \alpha(s[X_k]) = \alpha(q[X_k]) \cdot 
 u_{01k} \cdot u_{12k} \cdot u_{23k} \cdots u_{(k-1)kk} \cdot \alpha(q_m[X_k]s'[X_k])
\end{eqnarray*}
The type of $s[X_k \setminus \set{x_k}]$ is decomposed the same way,
only $u_{(k-1)kk}$ is replaced by $w_{(k-1)kk}$. Let
$h=\alpha(q_m[X_k]s'[X_k])$ and notice that if $q_m$ is a tree-context then $h$
is a tree-type. Therefore, the lemma will follow if
\begin{eqnarray*}
 u_{01k} \cdot u_{12k} \cdot u_{23k} \cdots u_{(k-1)kk} \cdot h = 
 u_{01k} \cdot u_{12k} \cdot u_{23k} \cdots w_{(k-1)kk} \cdot h \ .
\end{eqnarray*}
Since the fractal is monochromatic, and since $k=\omega+1$, the above becomes
\begin{eqnarray*}
 u_{01k}^\omega \cdot u_{(k-1)kk} \cdot h =   u_{01k}^\omega \cdot w_{(k-1)kk} \cdot h \ .
\end{eqnarray*}

By~(\ref{eq:lca-prop-decompose1}) and monochromaticity, we have
\begin{eqnarray}\label{eq-ccapiece}
 w_{(k-1)kk}\ , \  u_{(k-1)kk} \quad \ccapiece \quad
 u_{(k-1)k(k+1)}= u_{01k}\ ,
\end{eqnarray}
We now have two cases. If all the $q_i$ are tree-contexts, we conclude using
identity~(\ref{eq:lca-absorb}) which can be applied because
of~\eqref{eq-ccapiece}, and the fact that $h$ is then a tree-type
and~(\ref{eq:lca-prop-decompose2}).  If all the $q_i$ are contexts of the form
$\hole + f_i$, we conclude from~\eqref{eq-ccapiece} using identity~(\ref{eq:lca-forest-absorb}).
\end{proof}

\subsection{An equivalent set of identities.} 
\label{sec:an-equivalent-set-gca}
In this section, we give a set of identities that is equivalent to the
one used in Theorem~\ref{thm:main-lca}. The rationale is the same as
in Proposition~\ref{prop:other-eq}: we want to avoid the use of $v
\ccapiece w$ in the identities.

\begin{prop}\label{prop:other-set-lca}
 The conditions on the syntactic morphism stated in
 Theorem~\ref{thm:main-lca} are equivalent to the following
 equalities:

\begin{equation}\label{eq:lca-j-trivial-1}
  (uv)^\omega h = (uv)^\omega u h 
  \end{equation}
  whenever $h$ is a tree-type or empty,  and
 \begin{equation}\label{eq:lca-j-trivial-2}
  (uv)^\omega = v (uv)^\omega
   \end{equation}
   whenever $u$ and $v$ are tree-context-types, and
\begin{equation}
 \label{eq:new-simpler}
(u(\hole+vwh))^\omega g=(u(\hole+vwh))^\omega u(\hole+vh) g = (u(\hole+vh))(u(\hole+vwh))^\omega  g
\end{equation}
whenever $u$ is a tree-context-type or empty and $g,h$ are tree-types or empty.
\end{prop}

The rest of Section~\ref{sec:an-equivalent-set-gca} is devoted to showing
the above proposition.

It is immediate to see that identity~(\ref{eq:lca-absorb}) implies
identity~(\ref{eq:lca-j-trivial-2}) and that identity~(\ref{eq:lca-absorb})
implies identity~(\ref{eq:new-simpler}).  We now show that
identities~(\ref{eq:lca-absorb}) and~(\ref{eq:lca-forest-absorb}) imply
identity~(\ref{eq:lca-j-trivial-1}).  Let $u$ and $v$ be two context-types and
$h$ be a tree-type.  We want to show that $(uv)^\omega h= (uv)^\omega u h$.

We consider several cases.

\begin{iteMize}{$\bullet$}
\item In the first case we assume that $u=u_1u_2$ for some tree-context-type
$u_2$. In that case we have:
\begin{align*}
 (uv)^\omega h = (uv)^\omega(uv)^\omega(uv)^\omega
h=(u_1u_2vu_1u_2v)^\omega(u_1u_2v)^\omega h
= u_1(u_2vu_1)^{\omega-1}(u_2vu_1u_2vu_1)^\omega u_2vh
\end{align*}
Notice now that $u_2v \ccapiece u_2vu_1u_2vu_1$ and that $u_2vu_1u_2 \ccapiece
u_2vu_1u_2vu_1$. As $u_2$ is a tree-context-type, all the context-types involved are
tree-context-types and we can use identity~(\ref{eq:lca-absorb}) twice and replace
$u_2v$ by $u_2vu_1u_2$. This yields:
$$ (uv)^\omega h = u_1(u_2vu_1)^{\omega-1}(u_2vu_1u_2vu_1)^\omega u_2vu_1u_2h$$
And we have 
$$ (uv)^\omega h = (u_1u_2vu_1u_2vu_1u_2v)^\omega u_1u_2h$$
By idempotency, this yields the desired result:
$$(uv)^\omega h = (uv)^\omega u h $$

\item The second case, in which we assume that $v=v_1v_2$ for some tree-context-type $v_2,$
is treated similarly.
$$ (uv)^\omega h = (uv_1v_2)^\omega h= (uv_1v_2)^\omega(uv_1v_2)^\omega h$$
Therefore,
$$ (uv)^\omega h = uv_1(v_2uv_1)^{\omega-1}(v_2uv_1)^\omega v_2h$$

Notice now that $v_2 \ccapiece v_2uv_1$ and that $v_2u \ccapiece
v_2uv_1$. As $v_2$ is a tree-context-type, all the context-types involved are
tree-context-types and we can use identity~(\ref{eq:lca-absorb}) twice and replace
$v_2$ by $v_2u$. This yields:
$$ (uv)^\omega h = uv_1(v_2uv_1)^{\omega-1}(v_2uv_1)^\omega v_2uh$$
And we have 
$$ (uv)^\omega h = (uv)^\omega (uv)^\omega uh=(uv)^\omega uh$$

\item When none of the above cases works, we must have $u=f_1+\hole+f_2$
and $v=g_1+\hole+g_2$. In that case we have $(uv)^\omega h = \omega\cdot (f_1 + g_1)
+ h + \omega \cdot (g_2+f_2)$, and we conclude using
identity~(\ref{eq:lca-forest-absorb}) as $f_1 \ccapiece (f_1+g_1)$ and $f_2
\ccapiece (f_2+g_2)$.
\end{iteMize}\smallskip

\noindent We now consider the converse implication in
Proposition~\ref{prop:other-set-lca}.  Assume that
identities~(\ref{eq:lca-j-trivial-1})-(\ref{eq:new-simpler}) hold.  We
 show that identities~(\ref{eq:lca-absorb})
and~(\ref{eq:lca-forest-absorb}) are satisfied.

We first show the following lemma:
\begin{lemma}\label{lemma-hard}
  If $u$ is a tree-context-type, $v,w,w'$ are (not necessarily tree) context-types with 
  $w' \ccapiece w$, and $g,h$ are either tree-types or empty, then the following
  identity holds
\begin{equation}
 \label{eq:hard-won-lemma}
   (u(\hole + vwh))^\omega g = (u(\hole + vwh))^\omega u(\hole + vw'h) g
\end{equation}
\end{lemma}

Note that the identity~(\ref{eq:lca-forest-absorb}) is a direct
consequence of the above, by taking $u,v$ to be the empty context, and
$g,h$ to be the empty tree. We will also use the above lemma to
show~(\ref{eq:lca-absorb}), but this will require some more work.

\medskip

\begin{proof}
 The proof is by induction on the number of steps used to derive $w'
 \ccapiece w$.
\begin{iteMize}{$\bullet$}
\item Consider first the case when $w,w'$ can be decomposed as
     \begin{eqnarray*}
       w=w_1w_2 \qquad w'=w'_1w'_2 \qquad \qquad w'_1 \ccapiece w_1, w'_2 \ccapiece w_2
     \end{eqnarray*}
     Two applications of the induction assumption give us for all tree-type or
     empty $g$:
\begin{align}
(u(\hole + vw_1w_2h))^\omega g &= (u(\hole +
vw_1w_2h))^\omega u(\hole + vw_1w'_2h) g\label{eq-1}\\
(u(\hole + vw_1w'_2h))^\omega g &= (u(\hole + vw_1w'_2h))^\omega
u(\hole + vw'_1w_2'h) g\label{eq-2}
\end{align}
As $u$ is a tree-context-type we can iterate on~\eqref{eq-1} and then
apply~\eqref{eq-2} in order to derive:
\begin{equation}
(u(\hole + vw_1w_2h))^\omega g = (u(\hole + vw_1w_2h))^\omega (u(\hole +
vw_1w'_2h))^\omega u(\hole + vw'_1w_2'h) g  
\end{equation}
As $u$ is a tree-context-type, we can apply again~\eqref{eq-1} in the reverse
direction in order to derive the desired result.
\item Consider now the case when $w,w'$ can be decomposed as
     \begin{eqnarray*}
       w=w_1w_2w_3 \qquad w'=w'_1w'_3 \qquad w'_1 \ccapiece w_1, w'_3
       \ccapiece w_3
     \end{eqnarray*}
     with $w'_3$ a tree-context-type or empty. We first use the induction
     assumption to get

\begin{equation}\label{eq-3}
(u(\hole + vw_1 w_2 w_3 h))^\omega g = (u(\hole + vw_1 w_2w_3 h))^\omega
 u(\hole + vw_1w_2w'_3h) g
\end{equation}

By applying the identity~(\ref{eq:new-simpler}), we get for all tree-type or
empty $g$: 

\begin{equation}\label{eq-4}
(u(\hole + vw_1 w_2 w'_3 h))^\omega g = (u(\hole + vw_1 w_2w'_3h))^\omega
u(\hole + vw_1w'_3h) g
\end{equation}

Note that it is important here that $w'_3h$ is either a tree-context-type or
empty. Finally, we apply once again the induction assumption to get
     \begin{eqnarray}\label{eq-5}
       (u(\hole + vw_1 w'_3 h))^\omega g= (u(\hole + vw_1w'_3))^\omega u(\hole
       + vw'_1w'_3h) g
     \end{eqnarray}
As $u$ is a tree-context type, we can first iterate on~\eqref{eq-3}, then iterate on~\eqref{eq-4} and finally
applying~\eqref{eq-5} in order to get:
\begin{equation*}
  (u(\hole + vw_1 w_2 w_3 h))^\omega g = (u(\hole + vw_1 w_2w_3 h))^\omega
  (u(\hole + vw_1w_2w'_3h))^\omega (u(\hole + vw_1w'_3h))^\omega u(\hole
  + vw'_1w'_3h) g  
\end{equation*}
Because $u$ is a tree-context-type we can now apply~\eqref{eq-3}
and~\eqref{eq-4} in reverse to eliminate the inner products and obtain the
desired result.
\item Finally, consider the case when $w,w'$ can be decomposed as
     \begin{eqnarray*}
       w=\hole+ w_10 \qquad w'=\hole + w'_10 \qquad w'_1 \ccapiece w_1
     \end{eqnarray*}
     In this case, the identity becomes:
\begin{equation*}
(u(\hole + v'w_10))^\omega g= (u(\hole + v'w_10))^\omega u(\hole +
 v'w'_10) g
\end{equation*}
where $v'=v(h+\hole)$. The result now follows by induction assumption with
$w_1,w'_1$ in place of $w,w'$.
 \end{iteMize}\smallskip

\noindent We now claim that all cases have been considered. Assume first that either $w'$ or $w$
consists of several trees. Then, by the definition of  $\ccapiece$, $w'$ and $w$ can
be decomposed into smaller forests and we conclude using the first bullet. We
can thus assume that both $w$ and $w'$ are trees.
If $w'$ contains a node between its root and its hole then, by definition of $\ccapiece$,  we can decompose
$w$ and $w'$ and apply the second bullet. Similarly we can transform $w$ using
the first bullet until the third bullet can be applied.
\end{proof}

We now derive the first part of identity~(\ref{eq:lca-absorb}). Let
$u$, $v$ be tree-context-types such that $v \ccapiece u$, and let $h$ be a
tree-type.  We show by induction on $v$ that $u^\omega h=u^\omega vh$. If $v=v_1v_2$ where
both $v_1$ and $v_2$ are tree-context-types then we consider $v_2$ first
and $v_1$ next:
\begin{eqnarray*}
 u^\omega h = u^\omega v_2 h = u^\omega v_1 v_2 h\ . 
\end{eqnarray*}
It is important here that $v_2 h$ is a tree-type.  

Therefore it is enough to consider the case where $v$ is of the form
$\alpha(a)(\hole+f)$ for some letter $a$ and some forest-type $f$. In the
sequel we write $a$ instead of $\alpha(a)$ in order to improve
readability. From $v\ccapiece u$ we get $u=u_1 a(\hole + g) u_2$ where
$u_1$ and $u_2$ are tree-context-types and $f \ccapiece g$.  Then we have from
identity~(\ref{eq:lca-j-trivial-1}) for any tree-type $h$:
\begin{align*}
u^\omega h &= (u_1a(\hole+g)u_2)^\omega h = u^\omega u_1 a(\hole + g) h\\
u^\omega h &= (u_1a(\hole+g)u_2)^\omega h = u^\omega u_1 h
\end{align*}
and therefore, as $a(\hole + g) h$ is a tree-type we get for any tree-type $h$:
\begin{align}\label{eq-a}
u^\omega h = u^\omega a(\hole+g) h
\end{align}
Iterating on~\eqref{eq-a} we get:
\begin{equation*}
u^\omega h = u^\omega a(\hole + g) h = u^\omega a(\hole + g)^\omega h \ .
\end{equation*}
It will therefore be enough to show 
\begin{eqnarray*}
 (a(\hole + g))^\omega h  = (a(\hole + g))^\omega a(\hole + f) h 
\end{eqnarray*}
for $f \ccapiece g$. This, however, is a consequence
of~(\ref{eq:hard-won-lemma}).

The second part of identity~(\ref{eq:lca-absorb}), $u^\omega=vu^\omega$, is
shown the same way using identity~(\ref{eq:lca-j-trivial-2}) instead of
identity~(\ref{eq:lca-j-trivial-1}) and building
on~\eqref{eq:hard-won-lemma-bis} below instead
of~\eqref{eq:hard-won-lemma}.
\begin{lemma}
  If $u$ is a tree-context-type, $v,w,w'$ are (not necessarily tree) context-types with 
  $w' \ccapiece w$, and $g,h$ are either tree-types or empty, then the following
  identity holds
\begin{eqnarray}\label{eq:hard-won-lemma-bis}
     (u(\hole + vwh))^\omega = (u(\hole +
     vw'h)(u(\hole + vwh))^\omega \ .
\end{eqnarray}
\end{lemma}
\begin{proof}
Identical to the proof of Lemma~\ref{lemma-hard}, applying the other side of identity~(\ref{eq:new-simpler}).  
\end{proof}

\section{Variations}

In this section we show that the techniques we developed in the previous
sections are fairly robust and can be adapted to many situations. We describe
some of them.

\subsection{Languages definable in $\Sigma_1$.}\label{sec-sigma1}
Here we treat the relatively simple case of languages defined by $\Sigma_1$
sentences (rather than boolean combinations of such formulas).  We will prove:

\begin{thm}\label{sigma_1_decidable} It is decidable whether a given regular
  forest language $L$ is definable by a $\Sigma_1(<,\orderfo)$ sentence.
\end{thm}

We will show how to do this using the syntactic forest algebra and syntactic
morphism, although this could be carried out just as well using an automaton
model.  The argument we give is based on an idea of Pin~\cite{pinordered}
concerning ordered monoids.

Let $L\subseteq H_{\A}$ be a regular forest language, and let
$\alpha_L:\A^{\Delta}\to(H_L,V_L)$ be its syntactic morphism.  We set
$X=\alpha_L(L)\subseteq H_L.$ Note that $L=\alpha_L^{-1}(X).$ For $h_1,h_2\in
H_L$ we define
$$h_1\leq_L^H h_2$$
if for all $v\in V_L,$ $vh_2\in X$ implies $vh_1\in X.$
Further, for $v_1, v_2\in V_2$ we define
$$v_1\leq_L^V v_2$$
if for all $h\in H_L,$ $v_1h\leq_L^H v_2h.$

\begin{prop}
The relations $\leq_L^H$ and $\leq_L^V$ are partial orders on $H_L$ and $V_L,$ respectively.  These orders are compatible with the algebra operations in the sense that whenever $h_1\leq_L^H h_2,$ $u_1\leq_L^V u_2,$ and $v_1\leq_L^V v_2,$ we have
$$v_1h_1\leq_L^H v_2h_2,$$
$$u_1v_1\leq_L^V u_2v_2.$$
\end{prop}

\begin{proof} This is straightforward from the definitions: Transitivity and
  reflexivity of $\leq_L^H$ are obvious.  To prove antisymmetry, suppose
  $h_1\leq_L^H h_2$ and $h_2\leq_L^H h_1.$ Let $s_1, s_2\in H_{\A}$ with
  $\alpha_L(s_i)=h_i.$ Let $p\in V_{\A}$ and set $v=\alpha_L(p).$ If $ps_2\in
  L$ then $vh_2=\alpha(ps_2)\in X,$ so $\alpha(ps_1)=vh_1\in X$ and thus
  $ps_1\in L.$ Likewise $ps_1\in L$ implies $ps_2\in L,$ so $s_1\sim_L s_2$ and
  thus $h_1=h_2.$

Transitivity and reflexivity of $\leq_L^V$ are likewise trivial, and antisymmetry follows from the antisymmetry of $\leq_L^H$ and the faithfulness of the action of  $V_L$ on $H_L.$  

For the multiplicative properties, let $h_i, u_i, v_i$ be as in the statement
of the Proposition.  If $v_2h_2\in X,$ then $v_2h_1\in X$ (since $h_1\leq_L^H
h_2$) and thus $v_1h_1\in X$ (since $v_1\leq_L^V v_2$). Thus $v_1h_1\leq_L^H
v_2h_2.$ Similarly $u_2v_2h\in X$ implies $u_1v_2h\in X$ (since $u_1 \leq_L^V
u_2$) and thus $u_1v_1h\in X$ (since $v_1 \leq_L^V v_2$) so $u_1u_2\leq_L^V
u_2v_2.$

\end{proof}

\begin{thm}  Let $L\subseteq H_{\A}$ be a regular forest language.  The following are equivalent:
\begin{iteMize}{$\bullet$} 
\item $L$ is definable by a $\Sigma_1(<,\orderfo)$ formula.
\item For all contexts $p$, $q$ and forests $t,$ \begin{eqnarray*}
  pt \in L \qquad \Rightarrow \qquad pqt \in L 
  \end{eqnarray*}
\item For all $v\in V_L,$ $v\leq_L^V \hole.$
\end{iteMize}
\end{thm}

\begin{proof} The first condition implies the second, because inserting new nodes in a forest does not change the $<$ or $\orderfo$ relation among the already existing nodes.  

To show that the second condition implies the first, we use a pumping argument:  Let $n=|H_L|.$  There exists $K>0$ such that any forest $s$ with at least $K$ nodes has a factorization
$$s=q_1q_2\cdots q_nt$$
for some forest $t,$ nonempty contexts $q_i$.  In particular, there is a factorization $s=pqt$ with $\alpha_L(t)=\alpha_L(qt).$  Thus a forest belongs to $L$ if and only if it is obtained by successive insertion of nodes starting with a forest in $L$ of size less than $K.$  We can write a $\Sigma_1$ sentence $\phi$ that describes all the relations among nodes of the forests of size less than $K$ that belong to $L,$ and thus this sentence defines $L.$

To show the equivalence of the second and third conditions, suppose the second condition holds.    We need to show $v\leq_L^V \hole$ for all $v\in V.$  This says that for every forest $s$ and every context $p,$ $s\in L$ implies $ps\in L,$ which follows from the second condition.  Conversely, suppose the third condition holds, and that $p,q$ are contexts and $t$ a forest with $pt\in L.$  Then $\alpha_L(pt)=\alpha_L(p)\hole \alpha_L(t)\in X.$  By the multiplicative properties of the partial order, $\alpha_L(p)\alpha_L(q)\alpha_L(t)\in X,$ and thus $pqt\in L.$\end{proof} 

Theorem~\ref{sigma_1_decidable} is an immediate corollary, since one can effectively compute the order $\leq_L^V$ given the syntactic algebra and syntactic morphism of $L.$

\subsection{Commutative languages} 

In this section we consider forest languages that are commutative,
\emph{i.e.,}~closed under rearranging siblings.

A forest $t'$ is called a \emph{reordering} of a forest $t$ if it is obtained
from $t$ by rearranging the order of siblings. In other words, reordering is
the least equivalence relation on forests that identifies all pairs of forests
of the form $p(s+t)$ and $p(t+s)$. A forest language is called
\emph{commutative} if it is closed under reordering. In other words, a forest
language is \emph{commutative} if and only if its syntactic forest algebra
satisfies the identity
\begin{eqnarray*}
  g+h=h+g\ .
\end{eqnarray*}

We say a forest $s$ is a \emph{commutative piece} of $t$, if $s$ is a
piece of some reordering of $t$. A forest language $L$ is called
\emph{commutative-piecewise testable} if for some $n \in \Nat$,
membership of $t$ in $L$ depends only on the set of commutative pieces of
$t$ that have no more than $n$ nodes. This definition also has a counterpart in
logic, by removing the forest-order from the signature. The following
proposition is immediate:

\begin{prop}
 A forest language is commutative-piecewise testable 
  iff it is definable by a Boolean combination of $\Sigma_1(<)$ formulas.
\end{prop}

If a language is commutative-piecewise testable, then it is clearly
commutative and piecewise testable (in the more powerful,
noncommutative, sense).
 Below we show that the converse implication is also
true:
\begin{thm}\label{thm:commutative}
  A forest language is commutative-piecewise testable if and only if
  it is commutative and piecewise testable.
\end{thm}

As piecewise testability is decidable, by Corollary~\ref{cor:decidable}, and
commutativity is obviously decidable, the theorem above implies decidability:

\begin{cor}
  It is decidable if a regular forest language is commutative-piecewise testable.
\end{cor}

Theorem~\ref{thm:commutative} follows quite easily from:
 \begin{lemma}\label{lemma-commutative}
   Let $n \in \Nat$. For $k$ sufficiently large, if two forests have
   the same commutative pieces of size at most $k$, then they can be
   both reordered so that the resulting forests have the same pieces of size at most
   $n$.
 \end{lemma}
 To see this, assume $L$ is a commutative and piecewise testable forest
 language. We need to show that there is a $k$ such that if $t$ and $s$ have
 the same commutative pieces of size $k$ then $t\in L$ iff $s\in L$. As $L$ is
 piecewise testable there exists an $n$ such that whenever $s$ and $t$ have the
 same pieces of size no more than $n$ then $t\in L$ iff $s\in L$. Let $k$ be the number
 given by Lemma~\ref{lemma-commutative} for that $n$. Assume now that $s$ and
 $t$ have the same commutative pieces of size $k$. By
 Lemma~\ref{lemma-commutative} they can be reordered into respectively $s'$ and
 $t'$ such that $s'$ and $t'$ have the same pieces of size $n$. Hence $s'\in L$
 iff $t' \in L$. But as $L$ is commutative this yields $s\in L$ iff $t\in L$ as
 desired.

 \begin{proof}[Proof of Lemma~\ref{lemma-commutative}]
   Let $P(s)$ be the set of pieces of $s$ that have size at most $n$.
   As in Lemma~\ref{lemma:eilenberg}, there is
   some $k$ such that any forest $s$ has a piece $t \piece s$ of size
   at most $k$ with $P(s)=P(t)$.  Let now $s_1,s_2$ be two forests
   with the same commutative pieces of size $k$. For $i=1,2$, consider
   the families
   \begin{eqnarray*}
     \mathcal P_i = \set{P(s'_i) : \mbox{$s'_i$ is a reordering of
         $s_i$}}\ .
   \end{eqnarray*}
   To prove the lemma, we need to show that the families $\mathcal P_1$ and
   $\mathcal P_2$ share a common element.  To this end, we show that for any $X
   \in \mathcal P_1$, there is some $Y \in \mathcal P_2$ with $X \subseteq Y$,
   and vice versa; in particular, the families share the same maximal elements.
   Let then $X=P(s'_1) \in \mathcal P_1$. By the choice of $k$, the forest $s'_1$
   has a piece $t$ of size at most $k$ with $P(t)=X$. Therefore $t$ is a
   commutative piece of $s_1$ of size $k$. By
   assumption, the forest $t$ is also a commutative piece of $s_2$ and
   therefore a piece of some reordering
   $s'_2$ of $s_2$. Hence $X \subseteq P(s'_2) \in \mathcal P_2$.
 \end{proof}

\medskip

Similarly we can define the notion of commutative-\lcapiece and
commutative-\lcapiecewise testable forest language. Using the same arguments as
above we can prove:

\begin{prop}
 A forest language is commutative-\lcapiecewise testable 
  iff it is definable by a Boolean combination of $\Sigma_1(\sqcap)$ formulas.
\end{prop}
\begin{thm}\label{thm:commutative-cca}
  A forest language is commutative-\lcapiecewise testable if and only if
  it is commutative and \lcapiecewise testable.
\end{thm}
\begin{cor}
  It is decidable if a regular forest language is commutative-\lcapiecewise testable.
\end{cor}

\subsection{Tree languages}
Our previous results were provided decidable characterizations for
\emph{forest} languages, and in fact the algebraic theory used here works best
when forests, rather than trees, are treated as the fundamental object.
Traditionally, though, interest has focused on trees rather than forests. Thus
we want to give a decidable characterization of the piecewise testable tree
languages or, equivalently, the sets of \emph{trees} that are definable by  
Boolean combinations of $\Sigma_1$ sentences.

For certain logics, like first-order logic over the descendant relation, or
first-order logic over successor, one can write a sentence that says ``this
forest is a tree'', and thus there is no need to treat tree and forest
languages separately.  For piecewise testability, we need to do something more,
since the set of all trees over a finite alphabet \A is not definable by a
Boolean combination of $\Sigma_1$ sentences over any of the predicates mentioned
in this paper.

We define a \emph{tree piecewise testable language} over a finite alphabet \A
to be the intersection of a piecewise testable forest language with the set of
all trees over \A. In other words this is the set of languages definable
by a Boolean combination of $\Sigma_1(<,\orderfo)$ formulas when we interpret these formulas in trees.  This is
preferable to defining a piecewise testable tree language to be a tree language
that is piecewise testable (as a forest language), since the latter definition
would only define tree languages that are either finite or contain only chains
(no branching). Moreover it would not correspond to the tree languages
definable by a Boolean combination of $\Sigma_1(<,\orderfo)$ formulas.  The
cases when the pieces are assumed to be commutative and/or take into account
\lca are defined analogously.

We will obtain our decidability result by a general method for translating
algebraic characterizations of classes of forest languages to characterizations
of the corresponding classes of tree languages. This method will apply to all
the cases we considered earlier: piecewise testable languages, \lcapiecewise testable languages, and their
commutative counterparts.

 First, suppose
$$\alpha:\A^{\Delta}\to (H,V)$$
is a surjective forest algebra morphism. Recall that we denote by
$H_{\A}$ the set of all forests of \A. Based on $\alpha$, we define an equivalence relation on
$H_{\A}$: We write $s\sim t$ if for all contexts $p$ such that $ps$ and $pt$
are both trees (this happens if $p$ is a tree-context or if $p$ is the empty
context and both $t$ and $s$ are trees) we have $\alpha(ps)=\alpha(pt)$. Notice
that if $s$ and $t$ are such that $\alpha(s)=\alpha(t)$ then $s\sim t$ and that
if $s$ and $t$ are both trees then $s \sim t$ implies $\alpha(s)=\alpha(t)$
(take $p=\hole$ in the definition of $\sim$). It
is clear that if $s\sim t$ then for any context $q,$ $qs\sim qt.$ Thus $\sim$
defines a forest algebra congruence on $\A^{\Delta}.$ Let
 $$\alpha':\A^{\Delta}\to (H',V')$$
 be the projection morphism onto the quotient by this congruence.
We call $\alpha'$ the \emph{tree reduction} of $\alpha.$ From the remark above
it follows that if $t$ and $s$ are both trees then $\alpha(s)=\alpha(t)$ iff $\alpha'(s)=\alpha'(t)$.

Let $\bf F$ be a family of forest languages over \A. We say that a set ${\mathcal
  F}$ of surjective forest algebra morphisms with domain $\A^{\Delta}$
\emph{characterizes $\bf F$} if a forest language $L$ belongs to $\bf F$ if and
only if $L$ is recognized by some morphism in ${\mathcal F}$.  We will further
assume that ${\mathcal F}$ is closed in the following sense: suppose $\alpha:\A^{\Delta}\to (H_1,V_1)$ belongs to ${\mathcal
  F},$ and $\beta:(H_1,V_1)\to (H_2,V_2)$ is a morphism onto a finite forest algebra. Then $\beta\alpha$ belongs to ${\mathcal F}.$
 
 \begin{thm}\label{thm:forests-to-trees} Let $\bf F$ and ${\mathcal F}$ be
as above, and let $L\subseteq H_{\A}$ be a set of trees.  Then there is a
forest language $K\in{\bf F}$ such that $L$ consists of all the trees
in $K$ if and only if the tree reduction of the syntactic morphism
$\alpha_L$ of $L$ belongs to ${\mathcal F}.$
 \end{thm}
 
 \begin{proof} Let $L$ be a tree language, $\alpha_L$ be its syntactic
   morphism and let $\alpha'_L:\A^{\Delta} \to (H'_L,V'_L)$ be its tree
   reduction.

   Assume first that there is a forest language $K$ such that $L$ consists of
   all the trees in $K$. Let $\alpha_K:\A^{\Delta} \to (H_K,V_K)$ be the
   syntactic morphism of $K$. By definition, $\alpha_K\in {\mathcal F}$. Fix $h\in
   H_K$ and let $t,s$ be forests such that $\alpha_K(t)=h=\alpha_K(s)$. We show
   that $\alpha'_L(s)=\alpha'_L(t)$. Suppose this is
   not the case. Then there exists a context $p$ such that $ps$ and $pt$ are
   both trees but $\alpha_L(ps)\neq\alpha_L(pt)$. By definition of $\alpha_L$
   this means that there exists a context $q$ such that $qps\in L$ but
   $qpt\not\in L$. From $qps\in L$ we know that $qps$ is a tree, hence, as $pt$
   is a tree, $qpt$ must also be a tree. By hypothesis this implies $qps\in K$
   but $qpt\not\in K$, contradicting $\alpha_K(t)=\alpha_K(s)$.
   
   Since $V_L'$ acts faithfully on $H_L',$ it follows that for any contexts $p$
   and $q,$ $\alpha_K(p)=\alpha_K(q)$ implies $\alpha'_L(p)=\alpha'_L(q).$ Thus
   $\alpha'_L=\beta\alpha_K$ for some morphism $\beta:(H_K,V_K)\to(H'_L,V'_L)$
   sending $h \in H_K$ to $\beta(h)=\alpha'_L(\alpha_K^{-1}(h))$. By hypothesis
   on ${\mathcal F}$ this implies that $\alpha'_L \in {\mathcal F}$.

   \medskip

   Conversely, suppose that $\alpha'_L$ belongs to ${\mathcal F}$. Let
   $X=\alpha'_L(L)$ and set $K=(\alpha'_L)^{-1}(X)$. From the hypothesis it
   follows that $K\in {\bf F}$. Assume that $t$ is a tree such that
   $\alpha'_L(t)\in X$. By definition of $X$, there is a tree $s\in L$ such
   that $\alpha'_L(s)=\alpha'_L(t)$. But as $\alpha'_L$ is the tree reduction
   of $\alpha_L$, we have $\alpha'_L(s)=\alpha'_L(t)$ implies
   $\alpha_L(s)=\alpha_L(t)$ and therefore $t\in L$. Hence $L$ is the set of
   trees of $K$.
 \end{proof}
 
 As a result we have:
 
\begin{cor}
  It is decidable if a regular tree language is tree (commutative)
  (cca-)piecewise testable.
\end{cor}
\begin{proof}
  We only give the proof for the piecewise testable case. The other cases are handled similarly.

  Let $\bf F$ be the family of piecewise testable forest languages over \A, and
  let ${\mathcal F}$ be the family of morphisms from $\A^{\Delta}$ onto finite
  forest algebras that satisfy the identities of Theorem~\ref{thm:main}. Notice
  that from Proposition~\ref{prop:other-eq} it follows that if $\alpha \in
  {\mathcal F}$ then $\beta\alpha\in {\mathcal F}$ for all onto morphism $\beta$. Hence
  $\bf F$ and ${\mathcal F}$ satisfy the hypothesis of
  Theorem~\ref{thm:forests-to-trees}.

  Consequently, a regular tree language $L$ is tree piecewise testable if and
  only if the tree reduction of $\alpha_L$ belongs to ${\mathcal F}.$ It remains to
  show that we can effectively compute the image of the tree reduction given
  $\alpha_L$. Consider $h \in H_L$ and notice that all the forests in
  $\alpha_L^{-1}(h)$ agree on $\alpha'_L$. Hence the procedure amounts to
  deciding which pairs of elements of the syntactic forest algebra are
  identified under the reduction, which we can do as long as we know which
  elements are images under $\alpha_L$ of trees. It is easy to see that if an
  element of $H_L$ is the image of a tree, then it is the image of a tree of
  depth at most $|V_L|$ in which each node has at most $|H_L|$ children, so we
  can effectively decide this as well.
\end{proof}

\newcommand{\hpiece}{\Vdash}
\subsection{Horizontal order}

We could also consider other natural predicates over forests.  Recall for
instance the definition of \emph{horizontal-order} with $x
\orderh y$ expresses the fact that $x$ is a sibling of $y$ occurring strictly
before $y$ in the forest-order.

Correspondingly we say that $s$ is a horizontal-piece of $t$, denoted $s
\hpiece t$, if there is an injective mapping from nodes of $s$ to nodes of $t$
that preserve the horizontal-order and the ancestor relationship.  An
equivalent definition is that the piece relation is the reflexive transitive
closure of the relation
\begin{align*}
\set{(pt, pat) : & \mbox{ $p$ is a context, $a$ is a node, $t$ is a forest or
    empty}\\
&\mbox{ and either $t$ is empty or $a$ does not have a sibling in $pat$}}
\end{align*}
From this notion of horizontal-piece we derive the notion of
horizontal-piecewise testability as expected and the very same proofs as in
Section~\ref{section-descendant} yield:

\begin{prop}
  A forest language is horizontal-piecewise testable iff it is definable by a
  Boolean combination of $\Sigma_1(\orderh,\orderfo)$ formulas.
\end{prop}

\begin{thm}
  A forest language is horizontal-piecewise testable if and only if its
  syntactic algebra satisfies the identity
\begin{equation}
  u^\omega v = u^\omega = v u^\omega 
\end{equation}
for all $u,v\in V_L$ such that $v\hpiece u.$
\end{thm}

This implies decidability of horizontal-piecewise testability and it would be
interesting to see what would be the corresponding equivalent set of identities
that does not make use of $\hpiece$, in the spirit of
Proposition~\ref{prop:other-eq}.

A straightforward adaptation of Section~\ref{section-cca} would also give a
decidable characterization of definability by a Boolean combination of
$\Sigma_1(<,\orderh,\sqcap)$.

\section{Conclusion/discussion}

Simon's theorem on ${\mathcal J}$-trivial monoids has emerged as one of the fundamental results in the algebraic theory of automata on words.  The principal contribution of the present paper has been to show that the use of forest algebras leads to a natural generalization of this theorem to trees and forests.  In proving this generalization we have introduced a number of new techniques that we believe will prove useful in the continuing development of the algebraic theory of tree automata.

Let us briefly indicate a few directions for further research. There is a purely algebraic formulation of Simon's theorem, stating that every finite ${\mathcal J}$-trivial monoid $M$  is the quotient of a finite monoid $N$ that admits a partial order compatible with the multiplication in $N$ and in which the identity is the maximum element.  Our new results have a similar formulation:  Every finite forest algebra satisfying the identities of Section 4 is the quotient of an algebra that admits compatible partial orders on both its horizontal and vertical components.  In fact, Straubing and Th\'erien~\cite{StraTher} have proved this order property of finite ${\mathcal J}$-trivial monoids directly, yielding a quite different proof of Simon's theorem.  It would be interesting to know whether such an argument is also possible for forest algebras.

In the word case, the boolean combinations of $\Sigma_1$-definable languages form the first level of hierarchy whose union is the first-order definable languages.  Little is known about the higher levels of this hierarchy, apart from the fact that it is strict.  Indeed, the problem of effectively characterizing the languages definable by boolean combinations of $\Sigma_2$-sentences has been open for many years.  In contrast, the first-order definable languages themselves constitute one of the first classes for which an effective algebraic characterization was given:  these are exactly the languages whose syntactic monoids are aperiodic.  (McNaughton and Papert~\cite{mcnaughton}.)  The corresponding problem for trees and forests, however, remains open:  We possess non-effective algebraic characterizations for the forest languages definable by first-order sentences over the ancestor relation, and for the related subclasses CTL and CTL* (see Boja\'nczyk, {\it et. al.} ~\cite{wreath}), but the problem of finding effective tests for membership of a language in any of these classes remains one of the greatest challenges in this work.

\end{document}